\documentclass[journal,final]{IEEEtran}

\usepackage{etoolbox}
\usepackage{graphicx}
\usepackage{amssymb}
\usepackage{amsmath}
\usepackage{multicol}
\usepackage{booktabs}
\usepackage{pdflscape} 
\usepackage{longtable}
\usepackage{cite}
\usepackage{url}
\usepackage[utf8]{inputenc}
\usepackage{amsthm}
\usepackage{subfigure}
\usepackage{caption}

\newtheorem{lemma}{Lemma}

\newtheorem{theorem}{Theorem}
\newtheorem{definition}{Definition}

\newcommand{\mypar}[1]{{\bf #1.}}

\newcommand{\coord}[1]{\text{\bf #1}}

\newcommand{\R}[0]{\mathbb{R}}
\newcommand{\C}[0]{{\mathbb C}}

\newcommand{\md}[0]{{\cal M}}
\newcommand{\alg}[0]{{\cal A}}
\newcommand{\bset}[0]{N}

\newcommand{\chr}[1]{\iota_{#1}}
\newcommand{\dps}{\displaystyle}

\newcommand{\DSFT}[0]{\operatorname{DSFT}}
\newcommand{\DSFTt}[1]{\DSFT^\text{\rm (#1)}}
\newcommand{\conv}[1]{\ast^{(#1)}}

\newcommand{\FT}[0]{\operatorname{FT}}
\newcommand{\FR}[0]{\operatorname{FR}}
\newcommand{\DFT}[0]{\operatorname{DFT}}
\newcommand{\WHT}[0]{\operatorname{WHT}}

\newcommand{\one}[0]{I}
\newcommand{\tensor}[0]{\otimes}

\newcommand{\matbottom}[0]{\left[\begin{smallmatrix}0&0\\1&1\end{smallmatrix}\right]}
\newcommand{\matleft}[0]{\left[\begin{smallmatrix}1&0\\1&0\end{smallmatrix}\right]}
\newcommand{\matcyc}[0]{\left[\begin{smallmatrix}0&1\\1&0\end{smallmatrix}\right]}

\newcommand{\pmi}[0]{\phantom{-}}
\newcommand{\ft}[1]{\widehat{#1}}
\newcommand{\sft}[2]{\widehat{#2}^{(#1)}}

\DeclareMathOperator*{\argmin}{arg\,min}

\newcommand{\ra}[1]{\renewcommand{\arraystretch}{#1}}

\begin{document}

\title{Discrete Signal Processing with Set Functions}
%
%
%
\author{Markus P{\"u}schel,~\IEEEmembership{Fellow,~IEEE,} and Chris Wendler,~\IEEEmembership{Student Member,~IEEE}%
\thanks{The authors are with the Department of Computer Science, ETH Zurich, Switzerland (email: pueschel@inf.ethz.ch, chris.wendler@inf.ethz.ch)}}
\maketitle

\begin{abstract}
Set functions are functions (or signals) indexed by the powerset (set of all subsets) of a finite set \boldmath$N$. They are fundamental and ubiquitous in many application domains and have been used, for example, to formally describe or quantify loss functions for semantic image segmentation, the informativeness of sensors in sensor networks the utility of sets of items in recommender systems, cooperative games in game theory, or bidders in combinatorial auctions. In particular, the subclass of submodular functions occurs in many optimization and machine learning problems.

In this paper, we derive discrete-set signal processing (SP), a novel shift-invariant linear signal processing framework for set functions. Discrete-set SP considers different notions of shift obtained from set union and difference operations. For each shift it provides associated notions of shift-invariant filters, convolution, Fourier transform, and frequency response.
We provide intuition for our framework using the concept of generalized coverage function that we define, identify multivariate mutual information as a special case of a discrete-set spectrum, and motivate frequency ordering.
Our work brings a new set of tools for analyzing and processing set functions, and, in particular, for dealing with their exponential nature. We show two prototypical applications and experiments: compression in submodular function optimization and sampling for preference elicitation in combinatorial auctions.
\end{abstract}


%
\IEEEpeerreviewmaketitle


\section{Introduction}

\IEEEPARstart{A}{t} the 
core of signal processing (SP) is a well-developed and powerful theory built on the concepts of time-invariant linear systems, convolution, Fourier transform, frequency response, sampling, and others \cite{Oppenheim:99}. Many of the SP techniques and systems invented over time build on these to solve tasks including coding, estimation, detection, compression, filtering, and others, on a diverse set of signals including audio, image, radar, geophysical, and many others \cite{Moura:09,Nebeker:98}. 

In recent years, the advent of big data has dramatically increased not only the size but also the variety of available data for digital processing. In particular, many types of data are inherently not indexed by time, or its separable extensions to 2D or 3D, but have index domains encoding other forms of relationships between data values. Thus, it is of great interest to port the core SP theory and concepts to these domains in a meaningful way to enable the power of SP.

\mypar{Graph signal processing} A prominent example is the emerging field of signal processing on graphs (graph SP) \cite{Sandryhaila:13,Shuman:13}. It is designed for signal values associated with the nodes of a directed or undirected graph, which is a common way to model data from social networks, infrastructure networks, molecular networks, and others. Leveraging spectral graph theory \cite{Godsil:01}, graph SP presents meaningful interpretations of the shift operator (adjacency matrix), shift-invariant systems, convolution, Fourier transform, sampling, and others \cite{Sandryhaila:13,Sandryhaila:14,Chen:15}. In \cite{Shuman:13}, the same is done based on the Laplacian \cite{Chung:97} instead of the adjacency matrix. The large number of follow-up works demonstrates the benefit of porting core SP theory to new domains. One prominent example are neural networks using graph convolution \cite{Bronstein:17}. In this overview, the authors show the benefits and motivate the need for, as they call it, other "Non-Euclidean convolutions."

\mypar{Set functions} In this paper we develop a novel SP theory and associated tools for discrete set functions. A set function is a signal indexed with the powerset (set of all subsets) of a finite set $N$, which means it is of the form $(s_A)_{A\subseteq N}$, $s_A\in\R$. One can think of $N$ as a ground set of objects; the set function assigns a value or cost to each of its subsets. The concept is fundamental and naturally appears in many applications across disciplines. The graph cut function associates with every subset of nodes in a weighted graph the sum of the edge weights that need to be cut to separate them \cite{Karci:10}. In recommender systems a set function can model the utility of every subset of items \cite{Tschiatschek:16}. In sensor networks, a set function can describe the informativeness of every subset of sensors \cite{Krause:08}. In auction design every bidder is modeled by a set function that assigns to every subset of goods to be auctioned the value for this bidder \cite{Parkes:06}. In game theory a cooperative game is modeled as a set function that assigns to every coalition (subset of a considered set of players) that can be formed the collective payoff gained \cite{Peleg:03}. Data on hypergraphs can also be viewed as (sparse) set functions that, for example, assign to every occurring hyperedge (subset of nodes) a value \cite{Bretto:13}.

As a consequence, the applications of set functions have been manifold. Examples include document summarization \cite{Lin:11}, marketing analysis \cite{Krause:14}, combinatorial auction design \cite{Parkes:06}, and recommender system design \cite{Tschiatschek:16}. Examples in signal and image processing include image segmentation \cite{Osokin:14}, 
compressive subsampling \cite{Baldassarre:16},
precoder design \cite{Zhu:12},
sparse sensing \cite{Coutino:18},
action/gesture recognition \cite{Wan:14},
motion segmentation \cite{Shen:18},
attribute selection \cite{Zheng:17},
and clustering \cite{Liu:14}.

In many of these applications, the goal is the minimization, maximization, or estimation of a set function. The main challenge across applications is its exponential size $2^{|N|}$, which means the set function is usually not available in its entirety but only chosen samples/queries $s_A$ obtained through an oracle, which itself can be costly. Thus, to make things tractable it is crucial to exploit any available structure of the set function. One common example that naturally occurs in most of the above applications is submodularity, a discrete form of concavity, which has given rise to a number of efficient algorithms \cite{Cunningham:85, Krause:14}.

Another line of work uses Fourier analysis. The powerset can be modeled as an undirected, $|N|$-dimensional hypercube with edges between sets that differ by one element. This model has made the Walsh-Hadamard transform (WHT) the classical choice of Fourier transform for set functions \cite{DeWolf:08}. A recent line of work aims to make set functions tractable by assuming and exploiting sparsity in the WHT (Fourier) domain \cite{Stobbe:12, Amrollahi:19}.

\mypar{Contribution: Discrete-set SP} In this paper we develop a novel core SP theory and associated Fourier transforms for signal processing with discrete set functions, called discrete-set SP, extending our preliminary work in \cite{Pueschel:18}. The derivation follows the algebraic signal processing theory (ASP) \cite{Pueschel:08a,Pueschel:06c}, a general framework for porting core SP theory to new index domains. In particular, ASP identifies the shift as the central, axiomatic concept from which all others can be derived. For example, \cite{Sandryhaila:13} applies ASP to derive graph SP from the adjacency matrix as shift.

We consider four variants of shifts, obtained from set union and difference operations. For each, we derive the associated signal models (in the sense defined by ASP) including shift-invariant filters, convolution, filtering, Fourier transform, frequency response, and others. Discrete-set SP is fundamentally different from graph SP by being separable: it inherently distinguishes between the neighbors of an index $A\subseteq N$, which have one element more or less, by providing $n = |N|$ shifts for each model. The technical details will become clear later. Our work complements the classical Fourier analysis based on the WHT that we include as fifth model, but is different in that it is built from directed shifts. We will discuss more closely related work in Section~\ref{related}.

We provide intuition on the meaning of spectrum and Fourier transform using the notion of generalized coverage function that we define. This shows that our notions of spectrum capture the concepts of complementarity and substitutability of the elements in $N$. In particular, we show that multivariate mutual information is a special case of spectrum.

Finally, we show two prototypical applications that demonstrate how discrete-set SP may help in approximating, estimating, or learning set functions: compression in submodular function optimization and sampling for preference elicitation in auctions.

\section{Set functions and their "$z$-transform"}\label{approach}

We will derive discrete-set SP using the general ideas and procedure provided by the algebraic signal processing theory (ASP) \cite{Pueschel:08a}. ASP identifies the shift (or shifts) as the axiomatic core concept of any linear SP framework. Once a shift (or shifts) is chosen, the rest follows: convolutions (or filters) are associated shift-invariant linear mappings\footnote{More precisely, the shift(s) generate the algebra of filters. An algebra is a vector space that is also a ring, i.e., it supports a multiplication operation.}  and the Fourier transform is obtained via their joint eigendecomposition. In our work the considered shifts will capture the particular structure of the powerset domain by using set difference and union operations.

One way of defining a shift is directly as linear operator on the signal vector as done, e.g., in \cite{Pueschel:08b} (space shift) or \cite{Sandryhaila:13} (graph adjacency shift). We will take a different, though mathematically equivalent, approach by defining the shift via an equivalent of the $z$-transform for set functions. To do so, we first introduce the concept of formal sums.

\mypar{Formal sums and vector spaces} Assume we are given $n$ symbols $A_i$, $0
\leq i< n$. The real vector space generated by the $A_i$ is the set of linear combinations
$V = \{\sum_{0\leq i< n}s_i A_i\mid s_i\in\R\}$ with $A_0, \dots, A_{n-1}$ as its basis. The elements of $V$ are {\em formal sums}, i.e., they cannot be evaluated. Addition and scalar multiplication are defined as expected, and satisfy all the vector space axioms:
$$
\sum_{0\leq i < n}s_i A_i + \sum_{0\leq i < n}t_i A_i = 
\sum_{0\leq i < n}(s_i+t_i) A_i
$$
and, for $\alpha\in\R$,
$$
\alpha\sum_{0\leq i < n}s_i A_i = \sum_{0\leq i < n}(\alpha s_i) A_i.
$$
There is no obvious definition of shift on the set of the $A_i$ since it carries no additional structure.

\mypar{Time signals and \boldmath$z$-transform} Recall that for a finite-duration discrete signal $\coord{s} = (s_i)_{0\leq i<n}$, the $z$-transform is given by
$$
\Phi:\ \coord{s}\mapsto s = s(x) = \sum_{0\leq i<n}s_ix^i,
$$
where we write $x = z^{-1}$. In other words $s$ is a polynomial, which is a special case of a formal sum with $A_i = x^i$.\footnote{Polynomials are of course also functions that can be evaluated for any $x$. This property is not needed in our derivations.} The structure of the $x^i$ naturally supports the time shift (or translation) corresponding to multiplying by $x$: $x\cdot x^i = x^{i+1}$, and thus the $x^i$ have been called time marks \cite{Kalman:69}. The extension to shifting by $x^k$ ($k$-fold shift) is obvious. In the finite case, $x^n = 1$, i.e., $x^n-1=0$, is assumed (see Fig.~\ref{time-shift}), which makes the shift cyclic and gives structure to the set of the $x^i$: it is a monoid and even a group\footnote{A monoid is a set supporting an associative operation (here $\cdot$) with a neutral element (here $x^0=1$). Invertibility makes it a group.}. Linear extension\footnote{A linear mapping on a finite-dimensional vector space is uniquely determined by its images on the basis.} of the shift from the basis of the $x^i$ to arbitrary signals in the $z$-domain yields:
\begin{align}
xs(x) &= \sum_{0\leq i<n}s_ix^{i+1}\text{ mod }(x^n-1) \nonumber \\
&= s_{n-1}x^0 + \sum_{1\leq i<n}s_{i-1}x^i. \label{timedelay}
\end{align}
Note that $x$ {\em advances} the $x^i$ which {\em delays} the signal. We will see later that the corresponding concepts for set functions will not coincide since our set shifts are not invertible.

\begin{figure}\centering
\includegraphics[scale=0.32]{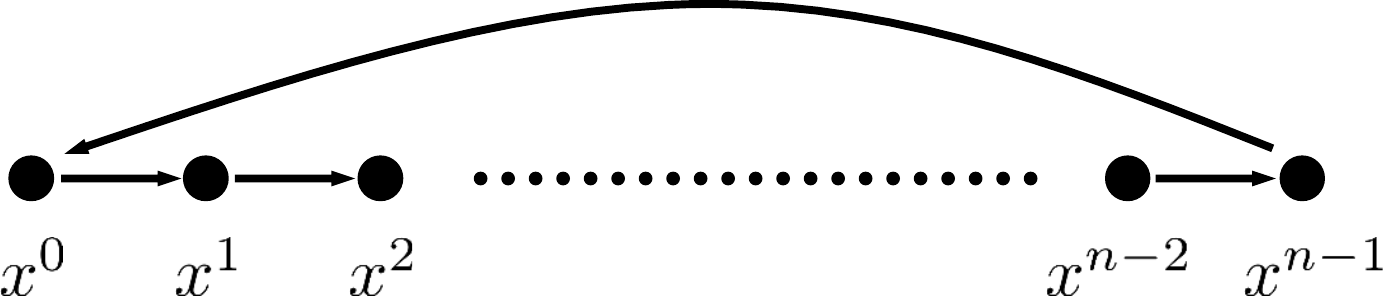}
\caption{Visualization of the (cyclic) time shift by $x$.}\label{time-shift}
\end{figure}

The associated notion of convolution extends \eqref{timedelay} to linear combinations $h(x) = \sum_{0 \leq k < n} h_k x^{k}$ of $k$-fold shifts, i.e., is given by $h(x)s(x)\text{ mod }x^n-1$, which is equivalent to the circular convolution of $\coord{h}$ and $\coord{s}$ as expected.

Next we formally introduce set functions and their equivalent of $z$-transform.

\mypar{Set functions and \boldmath$S$-transform} Given is a finite set $N = \{x_1,\dots,x_n\}$ of size $n$. A set function is a mapping on the powerset (set of all subsets, usually denoted with $2^N$) of $N$. In other words, it is a signal $\coord{s}$ whose index domain is not time but the powerset $2^N$. We consider this \emph{powerset signal} as a vector of length $2^n$, which thus has the form
$$
\coord{s} = (s_A)_{A\subseteq N}.
$$
In discrete-time SP, signals are naturally ordered by ascending time. For powerset signals we choose the lexicographic ordering on the Cartesian product $(\{\}, \{x_{n}\})\times\dots\times(\{\}, \{x_1\})$. For example, for $n = 3$ this yields the following ordering of subsets:
\begin{multline}\label{eq:ordersets}
\{\}, \{x_1\}, \{x_2\}, \{x_1,x_2\}, \\ 
\{x_3\}, \{x_1,x_3\}, \{x_2,x_3\}, \{x_1,x_2,x_3\}.
\end{multline}
The order is recursive: the first half contains all subsets not containing $x_3$ (again ordered lexicographically), the second half is the first half, each set augmented with $x_3$.
The powerset $2^N$ can be viewed as an $n$-dimensional hypercube (see Fig.~\ref{shift1} for $n=3$). A powerset signal associates values with its vertices.

\mypar{\boldmath$S$-transform} We define the equivalent of the $z$-transform for set functions, called $S$-transform ($S$ for set):
\begin{equation}\label{Strafo}
\Phi:\ \coord{s}\mapsto s = \sum_{A\subseteq N}s_A A.
\end{equation}
We will also refer to $s$ in the $S$-domain in \eqref{Strafo} as signal. $s$ is a formal sum, the set of which form a vector space as explained before. As the time marks $x^i$, the ``set marks" $A$ offer structure. For example, the union operation makes $2^N$ a monoid, and will provide one shift definition in the following. Following this idea we will consider four "natural" choices of shifts and thus derive four variants of discrete-set SP.

\section{Discrete-Set SP: Natural Shift}\label{dspset1}

The time shift advanced the time marks by 1: $x\cdot x^i = x^{i+1}$. On subsets, we choose as analogue to increase the sets by one element. Since there are $n$ ways of doing this, we define $n$ shift operators $x_i$ for each $x_i\in N$. We write this shift as multiplication:
\begin{equation}\label{natshift}
x_i\cdot A = A\cup \{x_i\},\quad A\subseteq N.
\end{equation}
By linear extension of \eqref{natshift} to arbitrary $s$ in \eqref{Strafo} we obtain
\begin{eqnarray}
x_i s & = & \sum_{A\subseteq N}s_A(A\cup\{x_i\}) \label{notinv} \\
 & = & \sum_{A\subseteq N, x_i\in A} (s_A + s_{A\setminus\{x_i\}}) A.\nonumber
\end{eqnarray}
The last sum is obtained by recognizing that the first sum only has summands for sets that contain $x_i$ and substitute $A$ for $A\cup \{x_i\}$. So the effect on the signal values is not the ``clean delay" $s_{A\setminus\{x_i\}}$ as one might have expected, and which would parallel \eqref{timedelay}, but $s_A + s_{A\setminus\{x_i\}}$ for $x_i\in A$ and 0 else. The reason is that the shift is not invertible: \eqref{notinv} lies in the $2^{n-1}$-dimensional subspace spanned by the sets that contain $x_i$.
Also note that the shift satisfies $x_i^2 = x_i$.

An example shift by $x_1$ is visualized in Fig.~\ref{shift1} for $n=3$, the shifts by $x_2,x_3$ operate analogously.

\begin{figure}\centering
\includegraphics[scale=0.32]{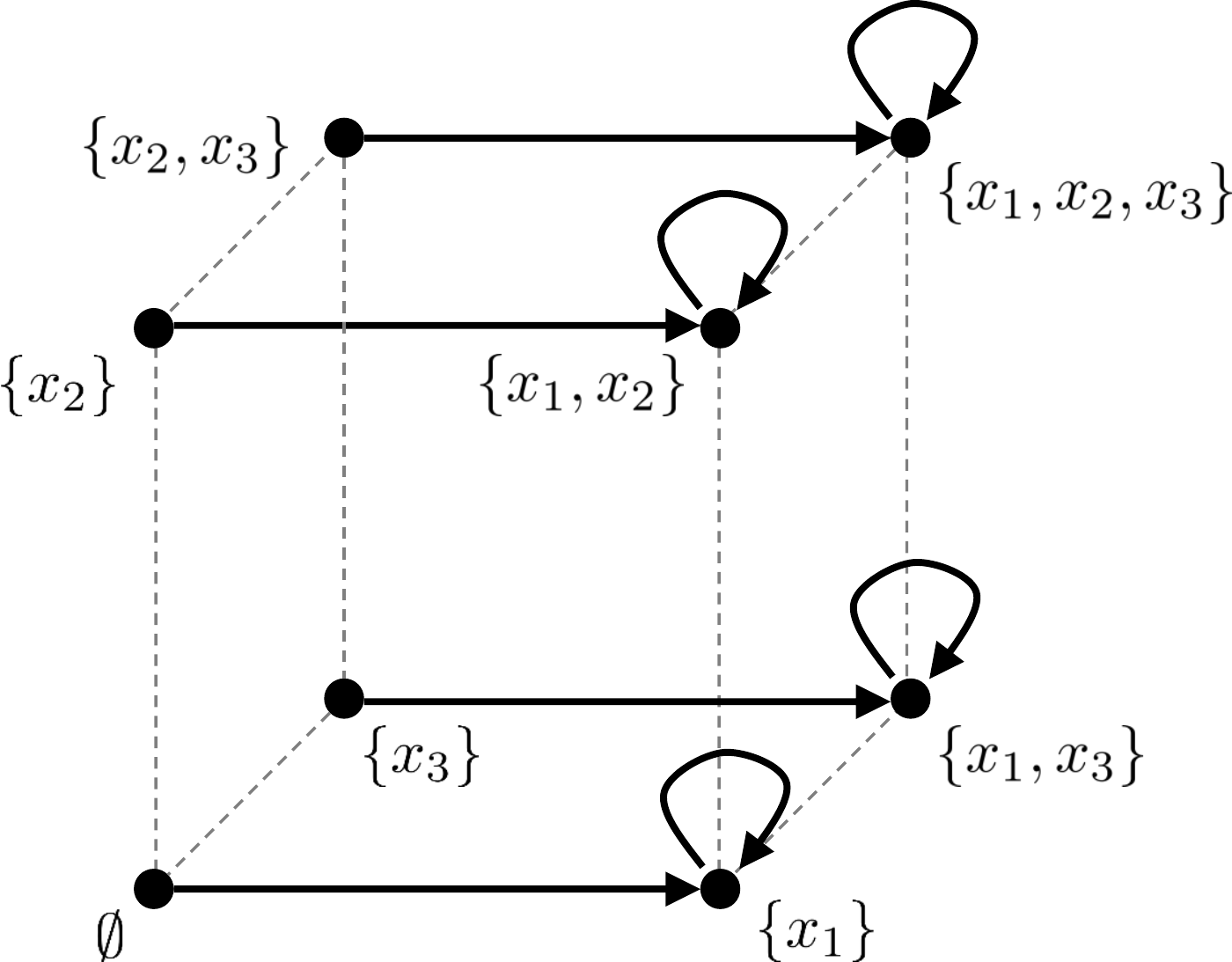}
\caption{Visualization of shift by $x_1$ in \eqref{natshift}.}\label{shift1}
\end{figure}

To obtain the matrix, denoted with $\phi(x_i)$, associated with the shift $x_ i$, we let it operate on the basis in our chosen order (see \eqref{eq:ordersets}). As an example we consider $n = 3$ and the shift $x_1$:
$$
\phi(x_1) = 
\left[
\scriptsize
\begin{array}{*{8}{c@{\quad}}}
0 & 0 & 0 & 0 & 0 & 0 & 0 & 0 \\
1 & 1 & 0 & 0 & 0 & 0 & 0 & 0 \\
0 & 0 & 0 & 0 & 0 & 0 & 0 & 0 \\
0 & 0 & 1 & 1 & 0 & 0 & 0 & 0 \\
0 & 0 & 0 & 0 & 0 & 0 & 0 & 0 \\
0 & 0 & 0 & 0 & 1 & 1 & 0 & 0 \\
0 & 0 & 0 & 0 & 0 & 0 & 0 & 0 \\
0 & 0 & 0 & 0 & 0 & 0 & 1 & 1 \\
\end{array}
\right] =
\one_{4}\tensor\matbottom.
$$
Here, $\one_m$ is the $m\times m$ identity matrix and $\tensor$ denotes the Kronecker product of matrices defined by $U\tensor V = [u_{k,\ell}V],\text{ for }U = [u_{k,\ell}]$, i.e., every entry of $U$ is multiplied by the entire matrix $V$. Again we see that $x_1$ is not invertible since $\phi(x_1)$ has rank 4.

In general,
\begin{equation}\label{natshiftmat}
\phi(x_i) = \one_{2^{n-i}}\tensor\matbottom\tensor\one_{2^{i-1}}.
\end{equation}

If matrices $U,U'$ have the same size and $V,V'$ have the same size, then $(U\tensor V)(U'\tensor V') = (UU'\tensor VV')$. Equation~\eqref{natshiftmat} shows that all shift matrices commute, as expected from \eqref{natshift}. We will diagonalize them simultaneously by the Fourier transform defined later.


\mypar{\boldmath$X$-fold shift for \boldmath $X\subseteq N$} We just defined a shift by an element of $N$. To work towards filtering we first need a consistent ``$X$-fold" shift for any subset $X\subseteq N$. This is done by shifting with all elements of $X$ in sequence. Since the union in \eqref{natshift} is commutative, the order does not matter. Formally, if $X = \{y_1,\dots,y_t\}\subseteq N$, then
\begin{equation}\label{Xnatshift}
X\cdot A = y_1(\cdots y_{t-1}(y_t\cdot A))\dots) = A\cup X.
\end{equation}
In particular, $\emptyset\cdot A = A$. By linear extension to signals $s$, we compute
\begin{eqnarray}
X\cdot s & = & \sum_{A\subseteq N}s_A (A\cup X) \nonumber \\ 
 & = & \sum_{X\subseteq A\subseteq N}\Bigl(\sum_{A\setminus X\subseteq B\subseteq A}s_B\Bigr)A. \label{eq:Xshift}
\end{eqnarray}
The last sum is obtained by observing that only summands for sets containing $X$ occur in the first sum, setting $A\cup X = A$, and collecting for each such $A$ all associated coefficients.

\eqref{Xnatshift} implies that the matrix representation of the shift by $X$ is given by the product 
\begin{equation}\label{Xnatshiftmat}
\phi(X) = \phi(y_1)\phi(y_2)\cdots\phi(y_t).
\end{equation}

\mypar{Filters} To make the filter space a vector space, a general filter is a linear combination of $X$-fold shifts and thus (in the $S$-domain) given by 
\begin{equation}\label{setfilter}
h = \sum_{X\subseteq N}h_X X
\end{equation}
and filtering (in the $S$-domain) becomes
$$
hs = \Big(\sum_{X\subseteq N}h_X X\Big)\Big(\sum_{A\subseteq N}s_A A\Big),
$$
which is again a signal. To compute it, we apply the distributivity law to reduce it to $X$-fold shifts and then use \eqref{eq:Xshift} to obtain
\begin{equation}\label{prod1}
hs = \sum_{A\subseteq N}\Big(\sum_{B\cup C= A}h_Bs_C\Big)A.
\end{equation}
Thus the expression in the parentheses is the associated notion of convolution on the coefficient vectors (sometimes called the covering product \cite{Bjorklund:07})
\begin{equation}\label{conv1}
(\coord{h}\conv{1}\coord{s})_A = \sum_{B\cup C = A}h_Bs_C.
\end{equation}
We refer to the superscript as ``type~1" since we will derive other types of convolutions based on different notions of shift later. 

The matrix representation of a filter $h$ in \eqref{setfilter} is given by 
$$
\phi(h) = \sum_{X\subseteq N}h_X\phi(X).
$$
As an example, we consider $n = 3$ and the filter $h = a\emptyset + b\{x_2\} + c\{x_1,x_3\} + d\{x_1,x_2,x_3\}$, $a,b,c,d\in\R$. Using \eqref{natshiftmat} and \eqref{Xnatshiftmat},
$$
\ra{1.3}
\begin{array}{l@{\ }l}
& \phi(h)\\
 = & a\one_8 + b(\one_2\tensor\matbottom\tensor\one_2) 
   + c(\one_4\tensor\matbottom)(\matbottom\tensor\one_4) \\
& +\ d(\matbottom\tensor\matbottom\tensor\matbottom)\\
 = &
a\one_8 + b(\one_2\tensor\matbottom\tensor\one_2) 
  + c(\matbottom\tensor\one_2\tensor\matbottom)\\
& +\ d(\one_4\tensor\matbottom)(\one_2\tensor\matbottom\tensor\one_2)(\matbottom\tensor\one_4)\\
 = &
\left[
\scriptsize
\ra{1.0}
\begin{array}{*{8}{c@{\quad}}}
 a  \\
  & a \\
 b & & a+b \\
  & b & & a+b \\
  & & & & a \\
 c & c & & & c & a+ c \\
  & & & & b & & a+ b \\
 d & d & c+d & c+d & d & b+d & c+d & a+b+c+d 
\end{array}
\right]
\end{array}
$$

\mypar{Shift invariance} Since $\cup$ is commutative, any shift by $x_i$ will commute with any filter, i.e., filters are shift-invariant. Formally, for all shifts $x_i\in N$, signals $s$, and filters $h$,
$$
x_i(hs) = h(x_is).
$$

\mypar{Fourier transform} The proper notion of Fourier transform should jointly diagonalize all filter matrices for which it is sufficient to diagonalize all shift matrices $\phi(x_i)$ in \eqref{natshiftmat}. This is possible since they commute. In fact, their special structure shows that this is achieved by a matrix of the form $T^{\tensor n} = T\tensor\dots\tensor T$, where $T$ diagonalizes $\matbottom$. Since
\begin{eqnarray}
\left[\begin{smallmatrix}0&\pmi 1\\1&-1\end{smallmatrix}\right]^{-1}
\left[\begin{smallmatrix}0&\pmi 0\\1&\pmi 1\end{smallmatrix}\right]
\left[\begin{smallmatrix}0&\pmi 1\\1&-1\end{smallmatrix}\right] & = &
\left[\begin{smallmatrix}1&\pmi 1\\1&\pmi 0\end{smallmatrix}\right]
\left[\begin{smallmatrix}0&\pmi 0\\1&\pmi 1\end{smallmatrix}\right]
\left[\begin{smallmatrix}1&\pmi 1\\1&\pmi 0\end{smallmatrix}\right]^{-1} \nonumber \\ & = &
\left[\begin{smallmatrix}1&\pmi 0\\0& \pmi 0\end{smallmatrix}\right], \label{eq:diagprop}
\end{eqnarray}
the discrete set Fourier transform (of type~1, since associated with \eqref{conv1}) is given by the matrix
\begin{equation}\label{dsft1kp}
\DSFTt{1}_{2^n} = 
\left[\begin{smallmatrix}1&1\\1&0\end{smallmatrix}\right]\tensor\dots\tensor
\left[\begin{smallmatrix}1&1\\1&0\end{smallmatrix}\right].
\end{equation}
Note that there is a degree of freedom in choosing $T$. We enforce the eigenvalue 1 in \eqref{eq:diagprop} to be first. Also, rows of the DSFT could be multiplied by $-1$.

We denote the spectrum of $\coord{s}$ with
$$
\sft{1}{\coord{s}} = \DSFTt{1}_{2^n}\coord{s}
$$
and \eqref{dsft1kp} shows that it can be computed with $n2^{n-1}$ additions.

The $\DSFTt{1}$ in \eqref{dsft1kp} can equivalently be represented in a closed form with a formula that computes every entry. The columns of $\DSFTt{1}$ are naturally indexed with $A\subseteq N$, as $\coord{s}$ is. Assume the same indexing for the rows, i.e., for the spectrum $\sft{1}{\coord{s}}$, in the same lexicographic order. Then
\begin{equation}\label{dsft1}
\DSFTt{1}_{2^n} = [\chr{A\cap B = \emptyset}(A, B)]_{B,A\subseteq N},
\end{equation}
where $\chr{}$ is the {\em indicator} function of the assertion in the subscript, i.e., in this case
$$
\chr{A\cap B = \emptyset}(A, B) = \begin{cases}1, & A\cap B = \emptyset,\\0, & \text{else.}\end{cases}
$$
By abuse of notation, we will often drop the arguments of a characteristic function. \eqref{dsft1} implies that the $B$th spectral component (or Fourier coefficient) of a signal $\coord{s}$ is computed as
\begin{equation}\label{dsft1closed}
\sft{1}{s}_B = \sum_{A\subseteq N, A\cap B = \emptyset}s_A.
\end{equation}
In other words, the spectrum is also a set function.

The closed forms for \eqref{dsft1} and for matrices occurring later in this paper are obtained using the following lemmas. Each assertion can be easily proven by induction over $n = |N|$.

\begin{lemma}\label{zeroloc} The following holds:
\begin{eqnarray*}
\begin{bmatrix} 0 & 1\\ 1 & 1\end{bmatrix}^{\tensor n} & = & [\chr{A\cup B=\bset}]_{A,B} = [\chr{N\setminus A\subseteq B}]_{A,B}\\
\begin{bmatrix} 1 & 0\\ 1 & 1\end{bmatrix}^{\tensor n} & = & [\chr{B\subseteq A}]_{A,B} \\
\begin{bmatrix} 1 & 1\\ 1 & 0\end{bmatrix}^{\tensor n} & = & [\chr{A\cap B=\emptyset}]_{A,B} = [\chr{B\subseteq N\setminus A}]_{A,B}\\
\begin{bmatrix} 1 & 1\\ 0 & 1\end{bmatrix}^{\tensor n} & = & [\chr{A\subseteq B}]_{A,B}
\end{eqnarray*}
\end{lemma}

\begin{lemma}\label{minusloc} The following holds:
\begin{eqnarray*}
\begin{bmatrix} -1 & 1\\ \pmi 1 & 1\end{bmatrix}^{\tensor n} & = & [(-1)^{|A\cup B|}]_{A,B} \\
\begin{bmatrix} 1 & -1\\ 1 & \pmi 1\end{bmatrix}^{\tensor n} & = & [(-1)^{|B\setminus A|}]_{A,B} \\
\begin{bmatrix} 1 & \pmi 1\\ 1 & -1\end{bmatrix}^{\tensor n} & = & [(-1)^{|A\cap B|}]_{A,B} \\
\begin{bmatrix} \pmi 1 & 1\\ -1 & 1\end{bmatrix}^{\tensor n} & = & [(-1)^{|A\setminus B|}]_{A,B}
\end{eqnarray*}
\end{lemma}
Note that in the lemmas $A$ is always the row index and $B$ the column index.

The lemmas can be combined to identify the closed form also in cases in which the $2\times 2$ matrix has one $0$, one $-1$, and two $1$s. For example, this yields a closed form for the inverse $\DSFTt{1}$: 
$$
(\DSFTt{1}_{2^n})^{-1} = 
\begin{bmatrix} 0 & \pmi 1\\ 1 & -1\end{bmatrix}^{\tensor n} = [(-1)^{|A\cap B|}\chr{A\cup B=\bset}]_{A,B},
$$
where Lemma~\ref{zeroloc} yields the nonzero pattern and the Lemma~\ref{minusloc} the minus-one pattern. Thus we also obtain a closed form for the Fourier basis, which consists of the columns of $(\DSFTt{1}_{2^n})^{-1}$. Namely, the $B$th Fourier basis vector $\coord{f}^B$ is the $B$th column:
$$
\coord{f}^B = ((-1)^{|A\cap B|}\chr{A\cup B=\bset})_{A\subseteq N}.
$$

\mypar{Frequency response} We first compute the frequency response of a shift by $x_i\in N$ at frequency $B$ using the $S$-domain. Let $B$ be fixed:
$$
x_i f^B = \sum_{A\subseteq N, A\cup B = N}(-1)^{|A\cap B|}(A\cup\{x_i\}).
$$
If $x_i\not\in B$, then $x_i$ is contained in every occurring $A$ (since $A\cup B = N$ implies $N \setminus B \subseteq A$) and thus $x_i f^B = f^B$. If $x_i \in B$, then every set $A\cup\{x_i\}$ occurs twice: once for an $A$ without $x_i$ that satisfies $A\cup B = N$ and once for the same $A$ joined with $x_i$. The intersection of these with $B$ differs in size by one and thus the associated summands cancel, yielding $x_i f^B = 0$. So the frequency response of the shift $x_i$ at the $B$th frequency is either 1 or 0, as expected from the last matrix in \eqref{eq:diagprop}.

Extending to a shift by $X\subseteq N$, using \eqref{Xnatshift}, yields
$$
Xf^B = \begin{cases}f^B, & \text{if }X\cap B = \emptyset, \\ 0, & \text{else,}\end{cases}
$$
and thus, by linear extension, we can compute the frequency response of an arbitrary filter $h$ at frequency $B$ through
$$
hf^B = \Bigl(\sum_{X\subseteq N}h_X X\Bigr)f^B = \Bigl(\sum_{X\subseteq N, X\cap B = \emptyset}h_X\Bigr)f^B.
$$
This shows that the frequency response is also computed with the $\DSFTt{1}$.

\mypar{Convolution theorem} The above yields the convolution theorem
$$
\sft{1}{\coord{h}\conv{1}\coord{s}} = \sft{1}{\coord{h}}\odot\sft{1}{\coord{s}},
$$
where $\odot$ denotes pointwise multiplication.

\section{Discrete-Set SP: Natural Delay}\label{dspset3}

The shift chosen in the previous section advanced the set marks in the $S$-domain but did not yield, what one could call the delay $s_{A\setminus\{x_i\}}$ of the signal, but instead $s_A + s_{A\setminus\{x_i\}}$ for $x_i\in A$ and 0 else. In this section we define, and build on, a shift that produces this delay.

We define a shift by $x_i\in N$ as
\begin{equation}\label{natdeldef}
x_i\cdot A = \begin{cases}A + A\cup\{x_i\}, & x_i\not\in A,\\ 0, & \text{else}.\end{cases}
\end{equation}
As before, we extend linearly to signals $s$ and compute
\begin{eqnarray*}
x\cdot s & = & \sum_{A\subseteq N, x_i\not\in A} s_A(A + A\cup \{x_i\}) \\
& = & \sum_{A\subseteq N, x_i\not\in A} s_A A + \sum_{A\subseteq N, x_i\in A} s_{A\setminus\{x_i\}} A\\
& = & \sum_{A\subseteq N} s_{A\setminus\{x_i\}} A,
\end{eqnarray*}
which is the desired set delay. For the second equality we split the sum and set $A = A\cup\{x_i\}$ in the second sum. For the third equality we used that for $x_i\not\in A$, $A\setminus\{x_i\} = A$. 

Fig.~\ref{shift3} visualizes a shift by $x_1$ for $n=3$. The sum in \eqref{natdeldef} yields two arrows that emanate from every set not containing $x_1$. Comparing to Fig.~\ref{shift1} reveals that these two shifts are, in a sense, dual to each other.

\begin{figure}\centering
\includegraphics[scale=0.32]{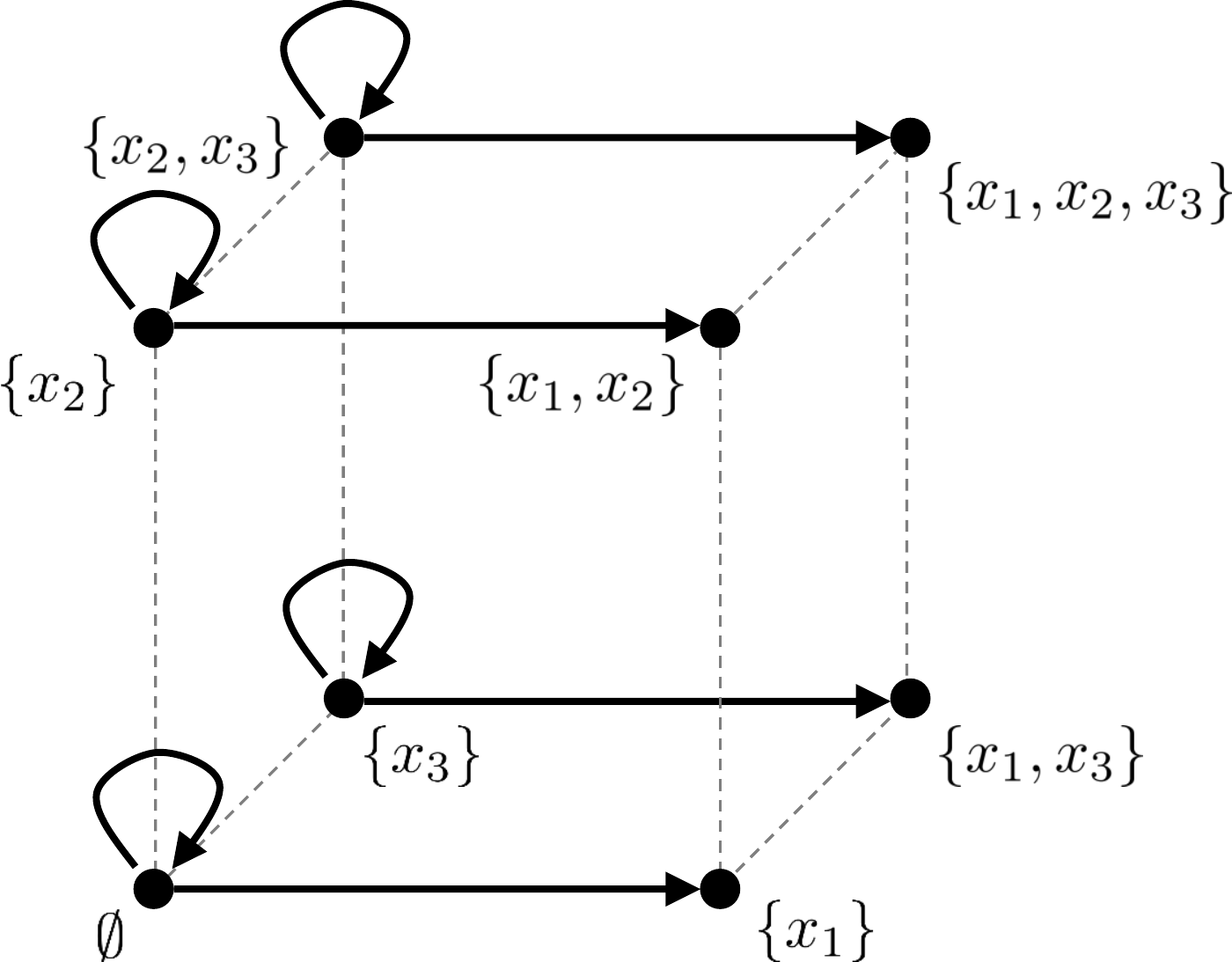}
\caption{Visualization of shift by $x_1$ in \eqref{natdeldef}.}\label{shift3}
\end{figure}

The associated matrix representation of the shift, by letting it operate on the subsets in the lexicographic order, now takes the form
\begin{equation}\label{eq:natdelmat}
\phi(x_i) = \one_{2^{n-i}}\tensor\matleft\tensor\one_{2^{i-1}}.
\end{equation}
As before, this also shows that the shifts commute.

\mypar{\boldmath{X}-fold shift for \boldmath$X\subseteq N$} As before, shifting by a set $X$ means shifting in sequence by all its elements, i.e., identifying $X$ with the product of its elements. This yields
$$
X\cdot s = \sum_{A\subseteq N}s_{A\setminus X} A.
$$

\mypar{Filters} Linearly extending the $X$-fold shifts to arbitrary $h = \sum_{X\subseteq N}h_X X$ yields the associated notion of filtering:
\begin{eqnarray}
hs & = & \sum_{X\subseteq N}h_X \Bigl(\sum_{A\subseteq N}s_{A\setminus X} A\Bigr) \nonumber \\
& = & \sum_{A\subseteq N}\Bigl(\sum_{X\subseteq N}h_Xs_{A\setminus X} \Bigr) A,\label{prod3}
\end{eqnarray}
which defines the convolution
$$
(\coord{h}\conv{3}\coord{s})_A = \sum_{X\subseteq N}h_Xs_{A\setminus X}.
$$
We refer to this convolution, and associated concepts later, as ``type~3."

\mypar{Shift invariance}
Since the shifts by $x_i$ commute and thus commute with shifts by any $X$, they also commute with any filter $h$, i.e., shift invariance holds.

\mypar{Fourier transform} We need to diagonalize all shift matrices in \eqref{eq:natdelmat}, i.e, diagonalize first $\matleft$:
\begin{eqnarray*}
\left[\begin{smallmatrix}1&\phantom{-}0\\1&-1\end{smallmatrix}\right]^{-1}
\left[\begin{smallmatrix}1&\phantom{-}0\\1&\phantom{-}0\end{smallmatrix}\right]
\left[\begin{smallmatrix}1&\phantom{-}0\\1&-1\end{smallmatrix}\right] & = &
\left[\begin{smallmatrix}1&\phantom{-}0\\1&-1\end{smallmatrix}\right]
\left[\begin{smallmatrix}1&\phantom{-}0\\1&\phantom{-}0\end{smallmatrix}\right]
\left[\begin{smallmatrix}1&\phantom{-}0\\1&-1\end{smallmatrix}\right]^{-1}\\
& = & \left[\begin{smallmatrix}1& \pmi 0\\0& \pmi 0\end{smallmatrix}\right].
\end{eqnarray*}
Thus the discrete set Fourier transform (of type~3) now takes the form
\begin{equation*}
\DSFTt{3}_{2^n} = 
\left[\begin{smallmatrix}1&\phantom{-}0\\1&-1\end{smallmatrix}\right]\tensor\dots\tensor
\left[\begin{smallmatrix}1&\phantom{-}0\\1&-1\end{smallmatrix}\right].
\end{equation*}
The complexity of computing the $\DSFTt{3}$ (of a powerset signal) is the same as for the $\DSFTt{1}$, namely $n2^{n-1}$ additions.

Using Lemmas~\ref{zeroloc} and \ref{minusloc}, the closed form is obtained as
$$
\DSFTt{3}_{2^n} = [(-1)^{|A|}\chr{A\subseteq B}]_{B,A\subseteq N}.
$$
Again, note that here the row index is $B$ and the column index $A$, accordingly the formulas from Lemmas~\ref{zeroloc} and \ref{minusloc} have to be adapted by swapping the indices.

The $\DSFTt{3}$ is self-inverse. Thus, the $B$th Fourier basis vector is given by
\begin{equation}\label{setfreq3}
\coord{f}^B = ((-1)^{|B|}\chr{B\subseteq A})_{A\subseteq N}.
\end{equation}

\mypar{Frequency response} Following the same steps as before, we compute first the frequency response of a single shift $x_i$. Shifting $\coord{f}^B$ by $x_i$ yields
$$
((-1)^{|B|}\chr{B\subseteq A\setminus\{x_i\}})_{A\subseteq N}.
$$
If $x_i\not\in B$, then $B\subseteq A\Leftrightarrow B\subseteq A\setminus\{x_i\}$, i.e., the shift does not change $\coord{f}^B$. If $x_i\in B$, then there is no $A$ satisfying $B\subseteq A\setminus\{x_i\}$ and the result is 0. In other words, the frequency response for shifts is the same as in Section~\ref{dspset1} and thus this also holds for arbitrary filters $h$. So the frequency response is computed with the $\DSFTt{1}$ (and not with the Fourier transform $\DSFTt{3}$ as one may have expected). 
This is not entirely surprising as it happens also with the discrete-space SP associated with the discrete cosine and sine transforms \cite{Pueschel:08b}. A deeper reason is the fundamental difference between Fourier transform and frequency response as explained in \cite{Pueschel:08a}.

\mypar{Convolution theorem} The above derivations yield
$$
\sft{3}{\coord{h}\conv{3}\coord{s}} = \sft{1}{\coord{h}}\odot\sft{3}{\coord{s}}.
$$

\section{Discrete-Set SP: Invertible Shift}

Both shifts defined in Sections~\ref{dspset1} and \ref{dspset3} lead to arguably natural translations of discrete-time SP to discrete-set SP but were both not invertible. While this does not prevent a meaningful notion of convolution and Fourier analysis it is still worth asking how to define an invertible shift. Doing so includes the prior work of Fourier analysis of set functions (e.g., \cite{DeWolf:08,Kahn:88}) using the well-known and well-studied Walsh-Hadamard transform in our framework. Since the mechanics of the derivation are the same as before and the results are known, we will be brief.

\mypar{Shift} An invertible shift can be defined as
\begin{equation}\label{invshiftdef}
x_i\cdot A = A\setminus\{x_i\}\cup\{x_i\}\setminus A = \begin{cases}A\cup \{x_i\}, & x_i\not\in A\\ A\setminus\{x_i\}, & x_i\in A\end{cases}
\end{equation}
The shift is visualized in Fig.~\ref{shift5}. It is undirected and the associated matrix representation becomes
\begin{equation}\label{eq:invdelmat}
\phi(x_i) = \one_{2^{n-i}}\tensor\matcyc\tensor\one_{2^{i-1}}.
\end{equation}
Note that the defining $2\times 2$ matrix is now invertible, as expected. Namely, $x_i^2 = 1$, i.e., $x_i^{-1} = x_i$.

\begin{figure}\centering
\includegraphics[scale=0.32]{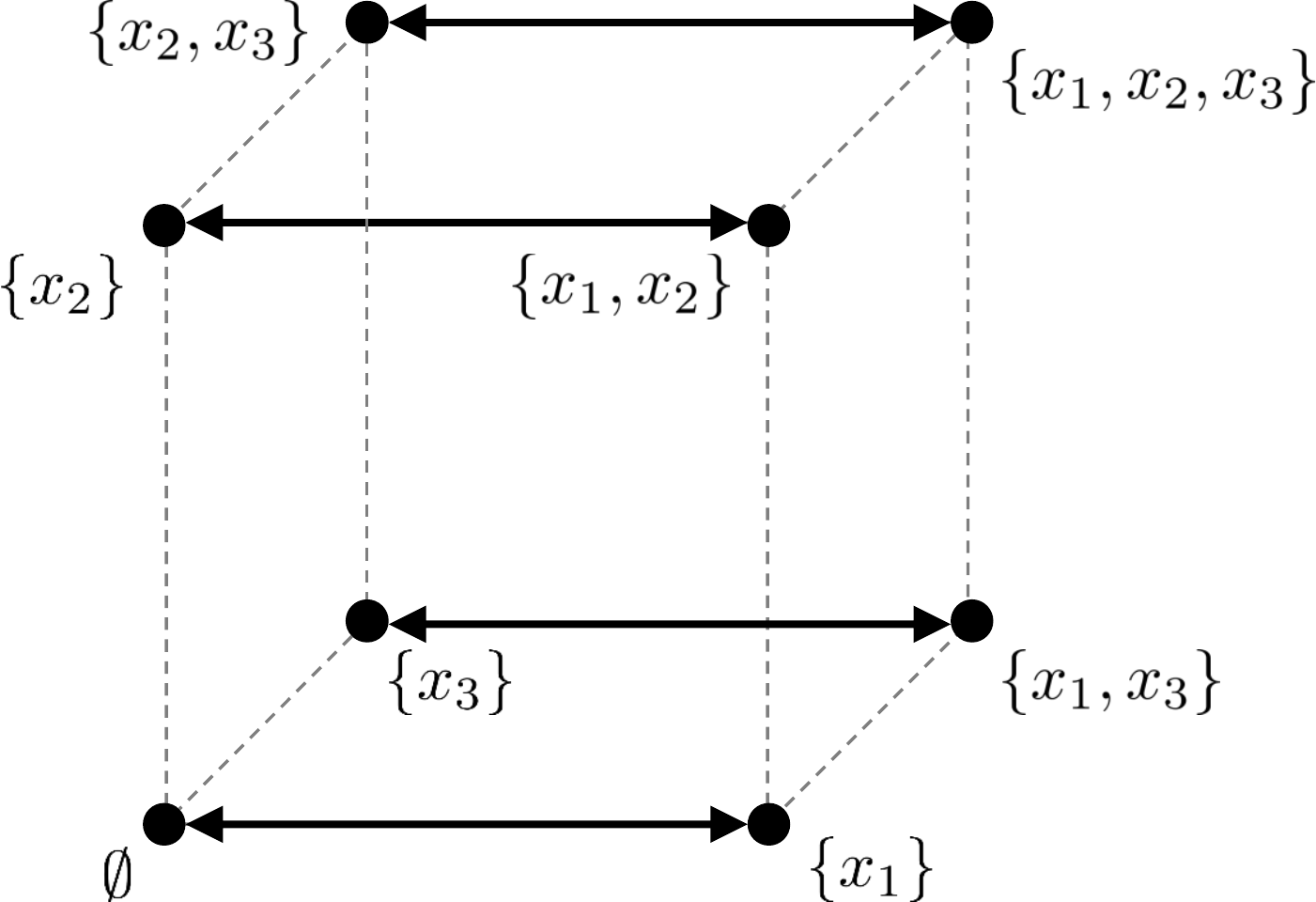}
\caption{Visualization of shift by $x_1$ in \eqref{invshiftdef}.}\label{shift5}
\end{figure}

\mypar{\boldmath$X$-fold shift for \boldmath $X\subseteq N$} Executing the above shifts for all elements in a set $X\subseteq N$ yields the so-called symmetric set difference
$$
X\cdot A = A\setminus X\cup X\setminus A.
$$

\mypar{Filters} Extending to linear combinations of $X$-fold shifts yields the associated convolution
$$
(\coord{h}\conv{5}\coord{s})_A = \sum_{X\subseteq N}h_Xs_{A\setminus X\cup X\setminus A},
$$
which we include as ``type~5" in this paper. Again, shift-invariance holds by construction.

\mypar{Fourier transform}
All shifts are now diagonalized by the Walsh-Hadamard transform \cite{Beauchamp:84,Walsh:23}:
\begin{eqnarray*}
\DSFTt{5}_{2^n} & = & \WHT_{2^n} \\
& = & 
\left[\begin{smallmatrix}1&\pmi 1\\1&-1\end{smallmatrix}\right]\tensor\dots\tensor
\left[\begin{smallmatrix}1&\pmi 1\\1&-1\end{smallmatrix}\right] \\
& = & [(-1)^{|A\cap B|}]_{A,B\subseteq N}.
\end{eqnarray*}
Equivalently, the above is a Kronecker product of $n$ $\DFT_2$. The inverse is thus computed as $\WHT_{2^n}^{-1} = (1/2)^n\WHT_{2^n}$, which yields the Fourier basis, for $B\subseteq N$:
\begin{equation}\label{setfreq5}
\coord{f}^B = (1/2)^n ((-1)^{|A\cap B|})_{A\subseteq N}.
\end{equation}
The $\WHT_{2^n}$ requires $n2^n$ additions.

\mypar{Frequency response and convolution theorem} The frequency response of a filter is also computed with the $\WHT$, which yields the convolution theorem
$$
\sft{5}{\coord{h}\conv{5}\coord{s}} = \sft{5}{\coord{h}}\odot\sft{5}{\coord{s}}.
$$

\section{Discrete-Set SP: All Models}

We presented three variants of discrete-set SP obtained from three different definitions of shift. We termed them type~1,3,5, respectively, where type~5 included the prior work on Fourier analysis of set functions using the WHT. We refer to the variants as signal models since it is up to a user to decide which one is appropriate for a given application. Also, each provides a signal model $(\alg,\md,\Phi)$ in the sense of ASP, namely an algebra $\alg$ of filters, an associated module of signals $\md$, and a generalized $z$-transform $\Phi$ \cite{Pueschel:08a}. 

We collect all concepts associated with these models in Tables~\ref{tab:smconv} and \ref{tab:smfreq}. The tables include two additional models 2 and 4 that we define next, followed by a discussion of the results and closely related work.

\newcommand{\mysp}[0]{\ }
\begin{table*}\centering
	\caption{Signal models for discrete-set SP: shift and convolution concepts. $q$ is any of the $x_i\in N$ and $Q\subseteq N$.\label{tab:smconv}}
	\tiny
	$
	\begin{array}{@{}l llll l l@{}}\toprule
	\text{model} & qA & \text{on signal} & QA & \text{on signal} & (\coord{h}\ast\coord{s})_A & \text{matrix for } q \\ \midrule
	1 & A\cup\{q\} & \begin{array}{@{}l@{\mysp}l@{}} s_A + s_{A\setminus\{q\}}, & q\in A\\ 0, & \text{else}\end{array} &
	A\cup Q & \begin{array}{@{}l@{\mysp}l@{}}\dps\sum_{A\setminus Q\subseteq B\subseteq A}s_B, & Q\subseteq A\\ 0, & \text{else}\end{array} &
	\dps\sum_{Q\cup B = A}h_Qs_B & \begin{bmatrix}0 & 0 \\ 1 & 1\end{bmatrix} \\[5mm]
	2 & A\setminus\{q\} & \begin{array}{@{}l@{\mysp}l} s_A + s_{A\cup\{q\}}, & q\not\in A\\ 0, & \text{else}\end{array} &
	A\setminus Q & \begin{array}{@{}l@{\mysp}l}\dps\sum_{B\subseteq Q}s_{A\cup B}, & Q\subseteq \bset\setminus A\\ 0, & \text{else}\end{array} &
	\dps\sum_{Q\subseteq\bset\setminus A}\dps\sum_{B\subseteq Q}h_Qs_{A\cup B}  & \begin{bmatrix}1 & 1 \\ 0 & 0\end{bmatrix} \\[5mm]
	3 & \begin{array}{@{}l@{\mysp}l} A + A\cup\{q\}, & q\not\in A\\ 0, & \text{else}\end{array} & s_{A\setminus\{q\}} &
	\begin{array}{@{}l@{\mysp}l}\dps\sum_{B\subseteq Q}A\cup B, & Q\subseteq \bset\setminus A\\ 0, & \text{else}\end{array} & 
	s_{A\setminus Q} &
	\dps\sum_{Q\subseteq\bset}h_Qs_{A\setminus Q} & \begin{bmatrix}1 & 0 \\ 1 & 0\end{bmatrix} \\[5mm]
	4 & \begin{array}{@{}l@{\mysp}l} A + A\setminus\{q\}, & q\in A\\ 0, & \text{else}\end{array} & s_{A\cup\{q\}} &
	\begin{array}{@{}l@{\mysp}l}\dps\sum_{B\subseteq Q}A\setminus B, & Q\subseteq A\\ 0, & \text{else}\end{array} & 
	s_{A\cup Q} &
	\dps\sum_{Q\subseteq\bset}h_Qs_{A\cup Q} & \begin{bmatrix}0 & 1 \\ 0 & 1\end{bmatrix} \\ \midrule
	5 & \begin{array}{@{}ll}\multicolumn{2}{@{}l}{A\setminus\{q\}\cup\{q\}\setminus A =}\\ A\cup\{q\}, & q\not\in A\\ A\setminus\{q\}, & \text{else}\end{array} &
	\begin{array}{@{}ll}\multicolumn{2}{@{}l}{s_{A\setminus\{q\}\cup\{q\}\setminus A} =}\\ s_{A\cup\{q\}}, & q\not\in A\\ s_{A\setminus\{q\}}, & \text{else}\end{array} &
	A\setminus Q\cup Q\setminus A & s_{A\setminus Q\cup Q\setminus A} &
	\dps\sum_{Q\subseteq\bset}h_Qs_{A\setminus Q\cup Q\setminus A} & \begin{bmatrix}0 & 1 \\ 1 & 0\end{bmatrix} \\ \bottomrule
	%
	%
	\end{array} 
	$
\end{table*}

\begin{table*}\centering
	\caption{Signal models for discrete-set SP: Frequency concepts. The Fourier transform (FT), its inverse, and the frequency response (FR) in matrix form are the $n$-fold Kronecker product of the $2\times 2$-matrix shown.\label{tab:smfreq}}
	\tiny
	$
	\begin{array}{@{}l llllll@{}}\toprule
	\text{model} & \text{matrix for } q & \FT \text{ (matrix)} & \FT^{-1} \text{ (matrix)} & \FT \text{ (sum)}: \widehat{s}_B = & \FT^{-1} \text{ (sum)}: {s}_A = & \FR \text{ (matrix)}\\ \midrule
	1 & \begin{bmatrix}0 & 0 \\ 1 & 1\end{bmatrix} & 
	\begin{bmatrix}1 & \pmi 1 \\ 1 & \pmi 0\end{bmatrix} & 
	\begin{bmatrix}0 & \phantom{-}1 \\ 1 & -1\end{bmatrix} & 
	\dps\sum_{A\subseteq N, A\cap B=\emptyset} s_A & 
	\dps\sum_{B\subseteq N, A\cup B=N} (-1)^{|A\cap B|}\widehat{s}_B & 
	\begin{bmatrix}1 & \pmi 1 \\ 1 & \pmi 0\end{bmatrix} \\[5mm]
	2 & \begin{bmatrix}1 & 1 \\ 0 & 0\end{bmatrix} & 
	\begin{bmatrix}1 & \phantom{-}1 \\ 0 & -1\end{bmatrix} &
	\begin{bmatrix}1 & \phantom{-}1 \\ 0 & -1\end{bmatrix} & 
	\dps\sum_{A\subseteq N, B\subseteq A} (-1)^{|A\cap B|}s_A & 
	\dps\sum_{B\subseteq N, A\subseteq B} (-1)^{|A\cap B|}\widehat{s}_B & 
	\begin{bmatrix}1 & \pmi 1 \\ 1 & \pmi 0\end{bmatrix}\\[5mm]
	3 & \begin{bmatrix}1 & 0 \\ 1 & 0\end{bmatrix} &
	\begin{bmatrix}1 & \phantom{-}0 \\ 1 & -1\end{bmatrix} &
	\begin{bmatrix}1 & \phantom{-}0 \\ 1 & -1\end{bmatrix} & 
	\dps\sum_{A\subseteq B} (-1)^{|A|}s_A & 
	\dps\sum_{B\subseteq A} (-1)^{|B|}\widehat{s}_B & 
	\begin{bmatrix}1 & \pmi 1 \\ 1 & \pmi 0\end{bmatrix}\\[5mm]
	4 & \begin{bmatrix}0 & 1 \\ 0 & 1\end{bmatrix} & 
	\begin{bmatrix}0 & \phantom{-}1 \\ 1 & -1\end{bmatrix} & 
	\begin{bmatrix}1 & \pmi 1 \\ 1 & \pmi 0\end{bmatrix} &
	\dps\sum_{A\subseteq N, A\cup B=N} (-1)^{|A\cap B|}s_A & 
	\dps\sum_{B\subseteq N, A\cap B=\emptyset} \widehat{s}_B & 
	\begin{bmatrix}1 & \pmi 1 \\ 1 & \pmi 0\end{bmatrix}\\ \midrule
	5 & \begin{bmatrix}0 & 1 \\ 1 & 0\end{bmatrix} & 
	\begin{bmatrix}1 & \phantom{-}1 \\ 1 & -1\end{bmatrix} &
	\frac 12\begin{bmatrix}1 & \phantom{-}1 \\ 1 & -1\end{bmatrix} &
	\dps\sum_{A\subseteq N}(-1)^{|A\cap B|} s_A &
	\left(\frac 12\right)^n \dps\sum_{B\subseteq N}(-1)^{|A\cap B|} \widehat{s}_B &
	\begin{bmatrix}1 & \phantom{-}1 \\ 1 & -1\end{bmatrix} \\ \bottomrule
	%
	%
	\end{array} 
	$
\end{table*}

\subsection{Shift by Subtracting Elements: Models 2 and 4} 

We complete discrete-set SP with two additional shift definitions that yield the models 2 and 4 in the tables. As analogue to \eqref{natshift} (model~1 in Section~\ref{dspset1}) we define a shift that decreases a set
$$
x_i\cdot A = A\setminus\{x_i\},
$$
and, as an analogue of model~3 (Section~\ref{dspset3}) we define a shift that yields a perfect advance of a signal, i.e., that has the effect
$$
s_A\mapsto s_{A\cup\{x_i\}}.
$$
The derivations of all concepts are analogous to before and we refer to the obtained models as type~2 and 4, respectively. The results are shown in the tables. Note that for all models 1--4 the frequency response is computed the same way, namely with the $\DSFTt{1}$, which thus plays a special role.

The last column in Table~\ref{tab:smconv} contains the $2\times 2$ matrices that define the shift matrices (as, for example, in \eqref{natshiftmat} and \eqref{eq:natdelmat}). We observe that the five variants are all possible matrices with two 1s and two 0s, except for the identity matrix which would yield a trivial model. So, in a sense, models 1--5 constitute one complete class.

\subsection{Discussion and Related Work}\label{related}

We discuss some of the salient aspects and properties of the discrete-set SP framework we derived and put it into the context of closely related prior work.

\mypar{Non-invertible shifts} The shifts for models 1--4 are not invertible, which is the main reason for having four variants. While this may seem to be a problem, our derivations show that all main SP concepts take meaningful forms. Also note that filters can still be invertible. In graph SP the Laplacian shift \cite{Shuman:13} is also not invertible and the adjacency shift \cite{Sandryhaila:13} not always.

\begin{figure}\centering
\subfigure[undirected]{\includegraphics[scale=0.27]{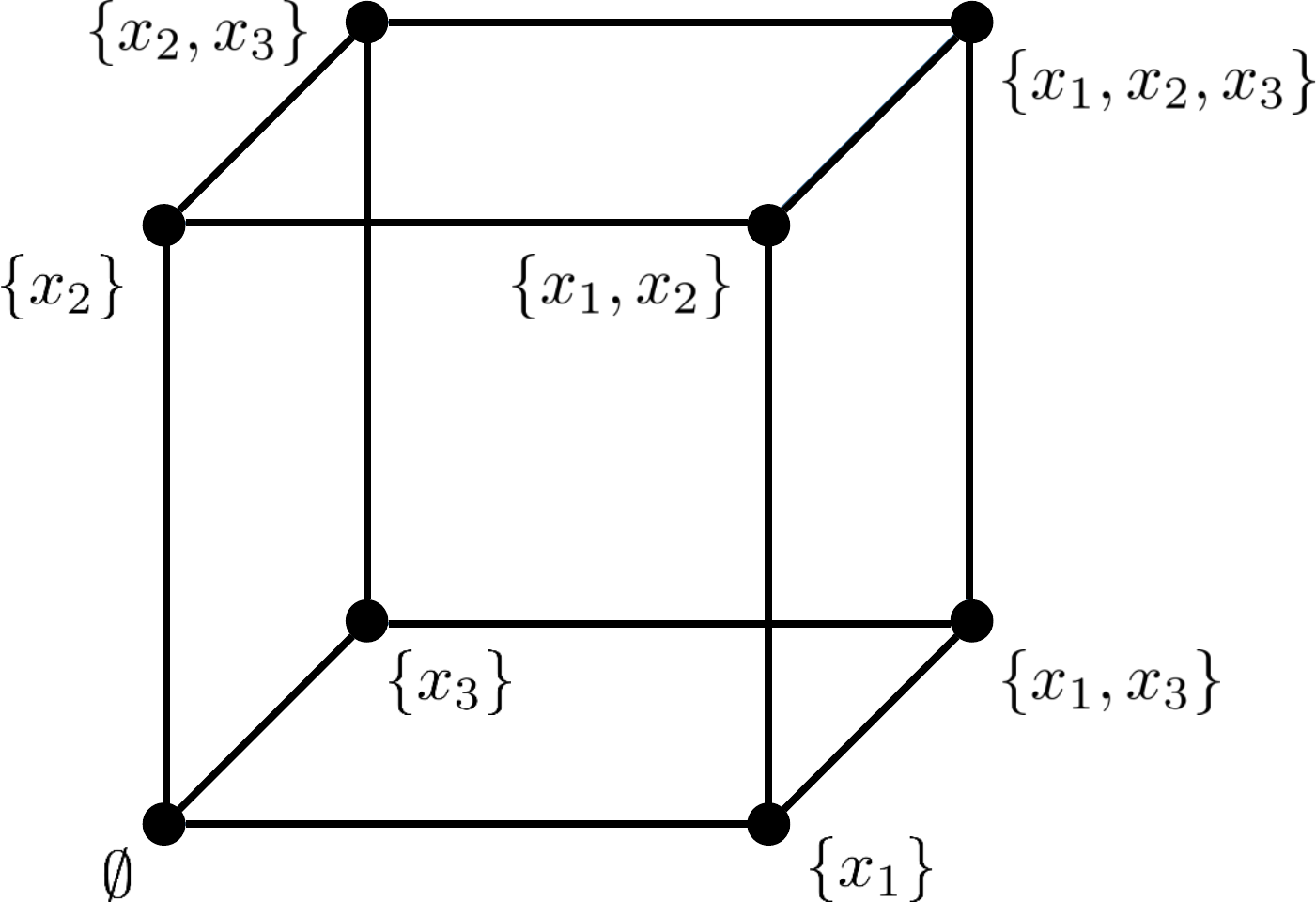}\label{hgraph}}
\subfigure[directed]{\includegraphics[scale=0.27]{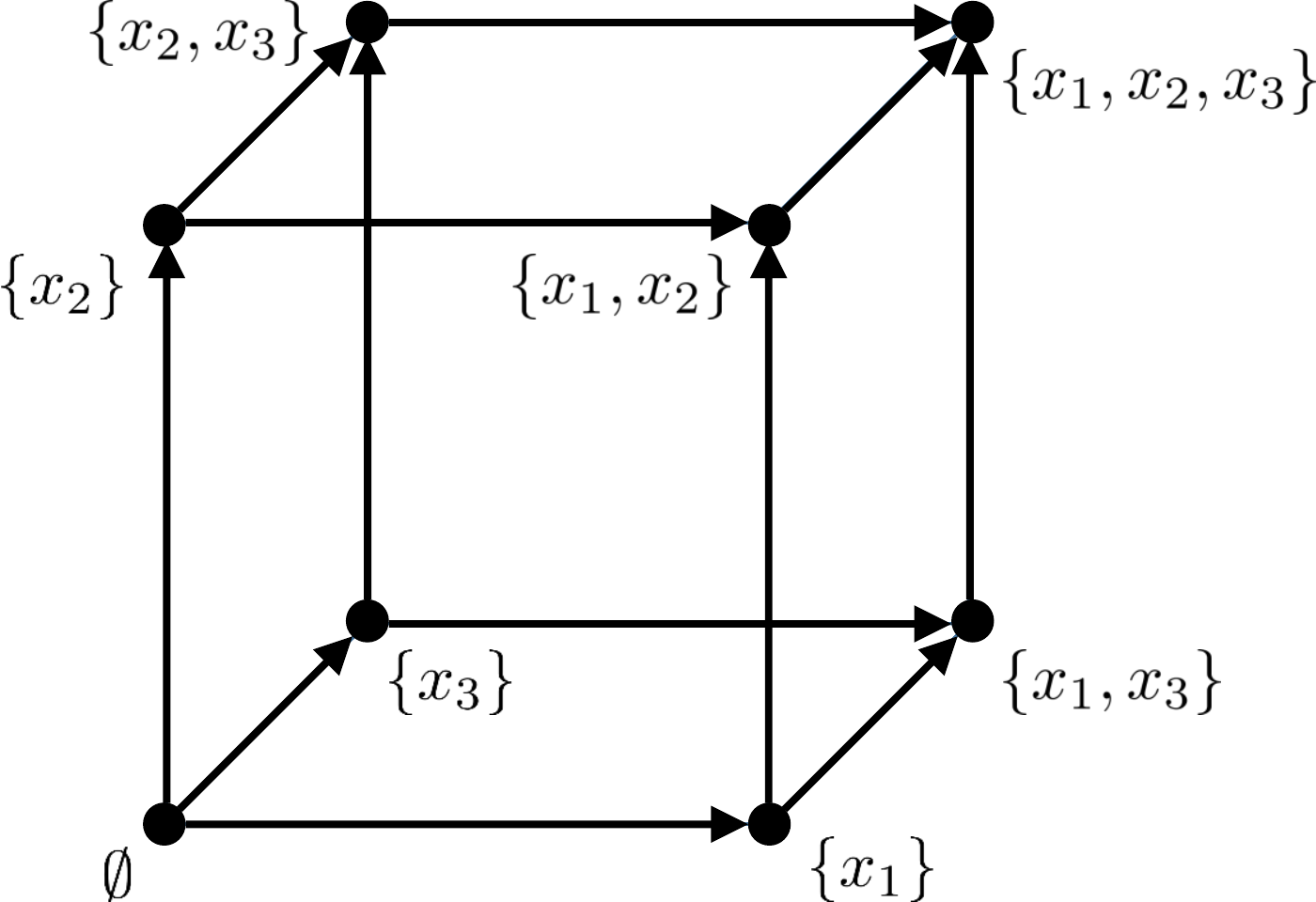}\label{dhgraph}}
\caption{The powerset domain as hypercube.}\label{gsp}
\end{figure}

\mypar{Difference to graph SP} We briefly explain the difference to graph SP based on adjacency \cite{Sandryhaila:12} or Laplacian matrix
\cite{Shuman:13}. The powerset domain can be naturally modeled as undirected or directed hypercube graph. In the undirected case  (Fig.~\ref{hgraph}), both the adjacency matrix $M$ and the Laplacian $nI_{2^n}-M$ are diagonalized by the WHT, which is the classical choice of Fourier transform for set functions. Namely, using our model~5, this follows from $M = \phi(x_1)+\dots+\phi(x_n)$ being the sum of all shift matrices. However, model~5 considers the $n$ shifts by $x_i$ separately, so they can be applied independently, i.e., model~5 is separable.
This is akin to images. Modeled as graphs, a shift (or translation) in one dimension would be equally applied in the other. But it is common to consider them separately and the $z$-transform has now two variables (shifts) $z_1,z_2$ \cite[pp.~174]{Dudgeon:83}. The result is a larger filter space.

However, our contribution is in the directed models. The directed hypercube graph in Fig.~\ref{dhgraph} is acyclic and thus the adjacency matrix $M$ has 0 as the only eigenvalue and is not diagonalizable, a known problem in digraph SP \cite[Sec.~III-A]{Ortega:18}. Here our models 1--4 not only provide separable SP frameworks as explained above, but also a proper (filter-diagonalizing) Fourier eigenbasis in each case. Interestingly, summing all shift matrices in the directed case yields the Laplacian-type matrices $D_\text{in} + M$ for model~1 and $D_\text{out}+M$ for model~3, where $D_\text{in}$ and $D_\text{out}$ are diagonal matrices with in- respectively outdegrees of the nodes. For models 2 and 4 the edges are reversed. Note that using these Laplacian-type matrices as starting point (i.e., shift) for graph SP, would pose another problem, besides the lack of separable filters. They have only $|N|+1$ different eigenvalues (already computed in Section~\ref{freqorder}): $|N/B|$, $B\subseteq N$, and thus large eigenspaces with no clear choice of basis. Our framework, based on $|N|$ shifts, yields a unique (up to scaling) basis.

Finally, we note that any finite, discrete, linear SP framework based on one shift is a form of graph SP (on a suitable graph) and vice-versa \cite[p.~56]{Pueschel:06c}. 

\mypar{Haar wavelet structure} The discrete set Fourier transforms have structure similar to a Haar wavelet that recursively applies basic low- or high-pass filters at different scales. However, other types of wavelets do not have an obvious interpretation on set functions and we could not get deeper insights from this observation. In Section~\ref{interpretation} we will provide a different form of intuition on the meaning of spectrum.

\mypar{Meet/join lattices} The powerset $2^N$ is a special case of a meet/join lattice, i.e., a set that is partially ordered (by $\subseteq$) with $A\cap B$ as the meet operation that computes the largest lower bound and $A\cup B$ as the join operation that computes the smallest upper bound of $A,B\subseteq N$. Our DSFTs are closely related to the Zeta- and Moebius transforms in lattice theory \cite{Rota:64}. \cite{Bjorklund:07} shows a convolution theorem for \eqref{conv1} (called covering product there) based on this connection. We have made first steps in generalizing this paper to signals indexed by arbitrary such lattices in \cite{Pueschel:19}. 

\mypar{Hypergraphs} A hypergraph (e.g., \cite{Bretto:13}) is a generalization of a graph that allows edges containing more than two vertices. It is given by $(V,E)$, where $V$ is the set of vertices and $E\subseteq 2^V$ are the hyperedges (usually, $\emptyset$ is excluded as hyperedge). Thus, the concepts of edge-weighted hypergraph and set function are equivalent. The role of vertices and edges can be exchanged to obtain a dual hypergraph, which means also the concepts of vertex-weighted hypergraph and set function are equivalent. In both cases the set function is usually very sparse as only few of the possible hyperedges are present.
An attempt to generalize SP methods from graphs to hypergraphs different from our work can be found in \cite{Barbarossa:16}. \cite{Barbarossa:20} developed an SP framework for simplicial complexes (a special class of hypergraphs that also supports a meet operation) using tools from topology.

\mypar{Other convolution and transforms} We note that other shifts and associated models are possible. For example, \cite{Bjorklund:07} studies and derives fast algorithms for a so-called subset convolution (and some of its variations): $(\coord{h}\ast\coord{s})_A = \sum_{Q\subseteq A}h_Q s_{A\setminus Q}$. In our framework it would be associated with the (non-diagonalizable) shift $qA = A\cup \{q\}$ if $q\not\in A$ and $=0$ else. The application in \cite{Bjorklund:07} are tighter complexity bounds for problems in theoretical computer science. For us, \cite{Bjorklund:07} provided the initial motivation for developing the work in this paper.

\mypar{Submodular functions} These constitute the subclass of set functions satisfying for all $A\subseteq N$, $x,y\in N$:
$$
s_{A\cup\{x\}} + s_{A\cup\{y\}} \geq s_{A\cup\{x,y\}} + s_A.
$$
Note that this definition connects nicely to our framework as it involves shifted versions of the set function. 

Examples of submodular functions include the entropy of subsets of random variables shown before in \eqref{jent}, graph cut functions, matroid rank functions, value functions in sensor placement, and many others. An overview of examples and applications in image segmentation, document summarization, marketing analysis, and others is given in \cite{Krause:14}. In many of these applications, the goal is the minimization or maximization of a submodular objective function, which is accessed through an evaluation oracle.

Reference \cite{Chakrabarty:12} introduces the W-transform as a tool for testing coverage functions, a subclass of submodular functions. The W-transform is equal to minus one times our DSFT for model~4. In Section~\ref{interpretation} we will generalize the concept of coverage function to provide intuition for the DSFT spectra.

References \cite{Stobbe:12, Amrollahi:19} use the WHT to learn submodular functions under the assumptions that the WHT spectrum is sparse. Both lines of work may benefit from the more general SP framework introduced in this paper.

\mypar{Game theory}
In game theory \cite{Peleg:03}, cooperative games  are equivalent to set functions that assign to every possible coalition (subset) from a set of a players the collective payoff gained. In this area, we find some of the mathematical concepts and transforms that we derive and define (e.g., \cite{Grabisch:00,Grabisch:16a}). Specifically, the Fourier basis vectors of our model~3 are (up to a scaling factor) called unanimity games, and those of the WHT parity games.

\mypar{Polynomial algebra view}
In discrete-time SP, circular convolution corresponds to the multiplication $h(x)s(x)\text{ mod }x^n-1$, i.e., in the polynomial algebra $\C[x]/(x^n-1)$ (see Section~\ref{approach}).
Similarly, our powerset convolutions can be expressed using polynomial algebras in $n$ variables. For example, \eqref{conv1} in model~1 corresponds to multiplication in $\R[x_1,\dots x_n]/\langle x_1^2-x_1,\dots,x_n^2-x_n\rangle$ with chosen basis polynomials $\prod_{x\in A} x$, $A\subseteq N$. This means that a subset $A$ is identified with the product of its elements. The other models are obtained with different choices of basis. We omit the details due to lack of space.

\section{Frequency Ordering and Filtering}\label{freqorder}

One question is how to order the spectrum of a powerset signal to obtain a notion of low and high frequencies. Since the spectrum is indexed by $B\subseteq N$, and the subsets are partially ordered by inclusion, it suggests to call frequencies with small $|B|$ low and high otherwise. For example, for models 3 and 4, using Table~\ref{tab:smfreq} (first column of inverse Fourier transform matrix), the Fourier basis vector $\coord{f}^\emptyset$ is constant 1 as for discrete-time SP.

Analogous to a moving average filter ($h = 1+x$ in the $z$-domain) one would assume that
$$
h = \emptyset + \sum_{i=1}^n \{x_i\}
$$
is a low pass filter. Its frequency response is the same for models 1--4, computed by $\DSFTt{1}$, and yields
$$
\widehat{h}_B = \sum_{A\cap B = \emptyset}h_A = 1 + |N\setminus B|.
$$
Indeed this shows that ``high" frequencies (large $|B|$) are attenuated compared to low ones.

We provide a deeper understanding on the meaning of spectrum, its ordering, and the Fourier transform next.

\section{Interpretation of DSFT Spectrum}\label{interpretation}

In this section we will give an intuitive interpretation of the DSFTs and the spectrum of a set function (or powerset signal), focusing on models 3 and 4. The key for doing so is in developing an alternative viewpoint towards set functions as generalized coverage functions that we define.

\subsection{Generalized Coverage Functions}\label{gencovsec}

For simplicity, we denote in the following the ground set as $N = \{1,\dots,n\}$. We start with defining a class of set functions.

\begin{definition}\label{gencovdef}
Let $\{S_1,\dots,S_n\}$ be a collection of subsets from some global set $U$: $S_i\subseteq U,\ 1\leq i\leq n$ and let $c\in\R$ be a constant. Further, we assume a weight function $w:\ U\rightarrow\R$ and its additive extension $\coord{w}$ to subsets: for $S\subseteq U$, $\coord{w}(S) = \sum_{u\in S} w(u)$. Then 
\begin{equation}\label{gencov}
\coord{s}:\ 2^N\rightarrow \R,\ A\mapsto c + \coord{w}\left(\bigcup_{i\in A}S_i\right)
\end{equation}
is a set function on $N$ called generalized coverage function. If $c = s_\emptyset = 0$, we call the function homogeneous.
\end{definition}
The setup is visualized in Fig.~\ref{coveragesit} and the definition of generalized coverage function in Fig.~\ref{coveragefct} for $c=0$. The shaded area signifies the weight of the respective set and can be negative or zero. The constant $c = s_\emptyset$ (the weight of the union of no sets $S_i$) can be viewed as the weight of $U\setminus\bigcup_{i\in N}S_i$. If this set is empty, then $c=0$.

\begin{figure}\centering
\includegraphics[scale=0.32]{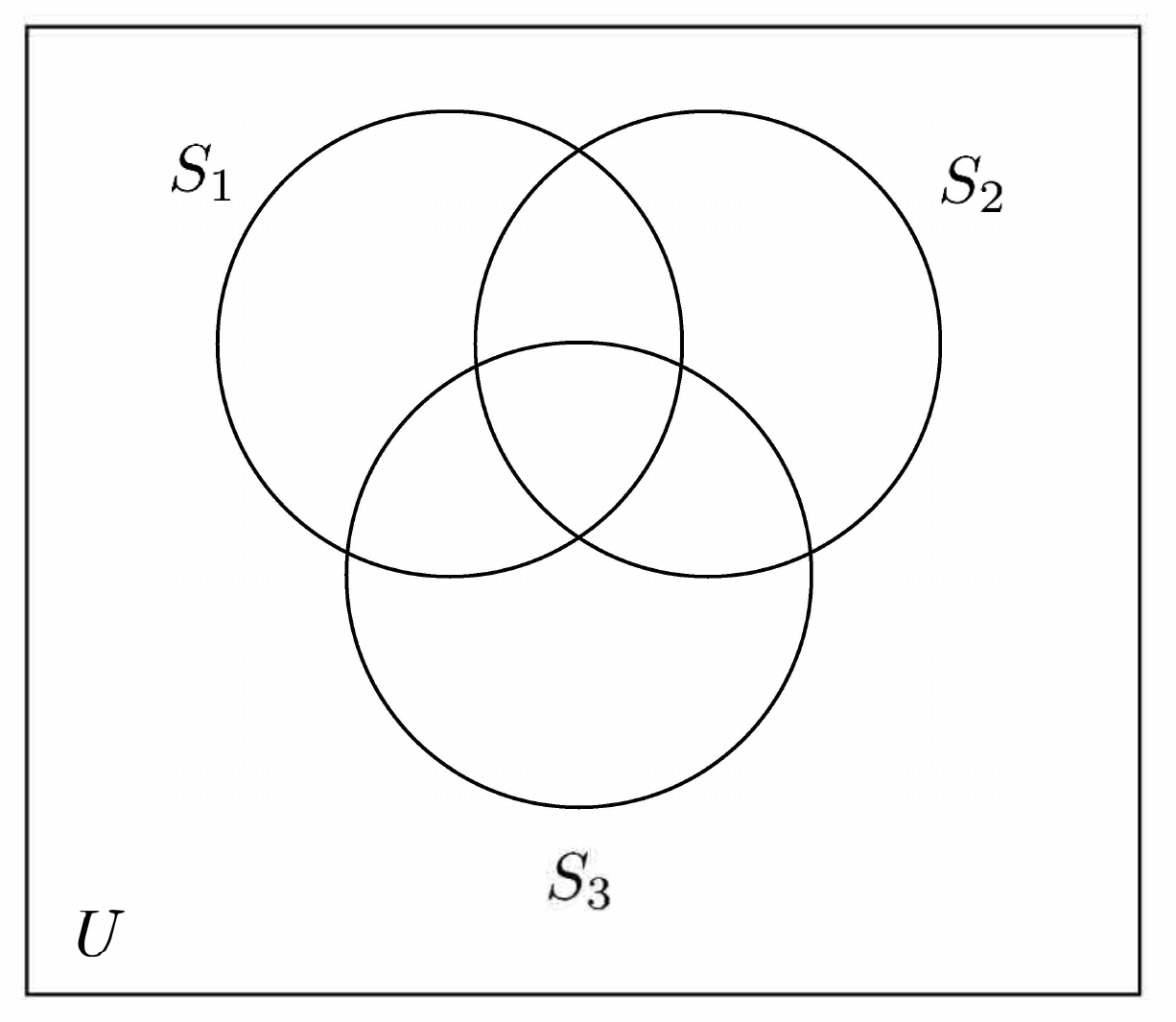}
\caption{Situation for a generalized coverage function for $N = \{1,2,3\}$.}\label{coveragesit}
\end{figure}

\begin{figure}\centering
\subfigure[generalized coverage function $\coord{s}$]{\includegraphics[scale=0.25]{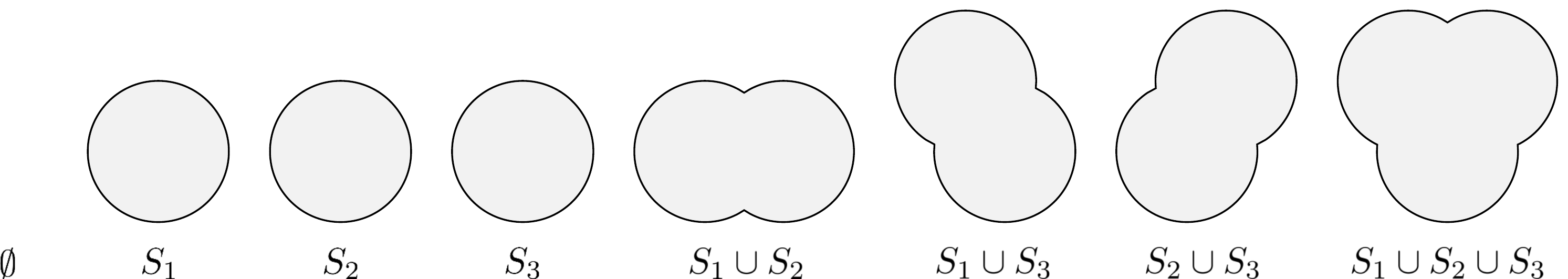}\label{coveragefct}}
\subfigure[$\sft{3}{\coord{s}}$]{\includegraphics[scale=0.25]{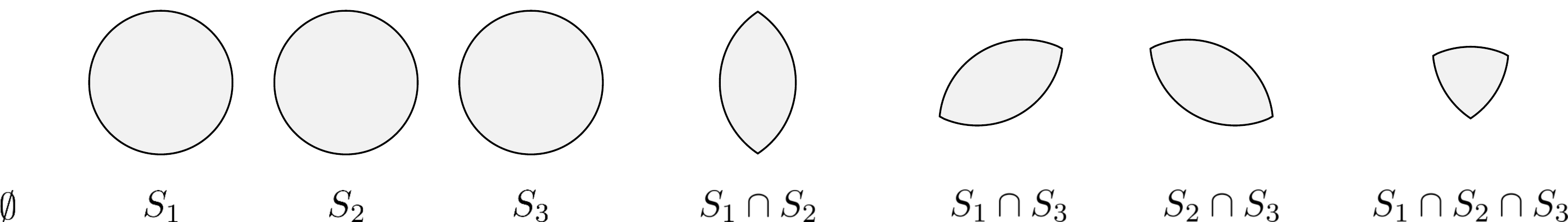}\label{dsft3vis}}
\subfigure[$\sft{4}{\coord{s}}$]{\includegraphics[scale=0.25]{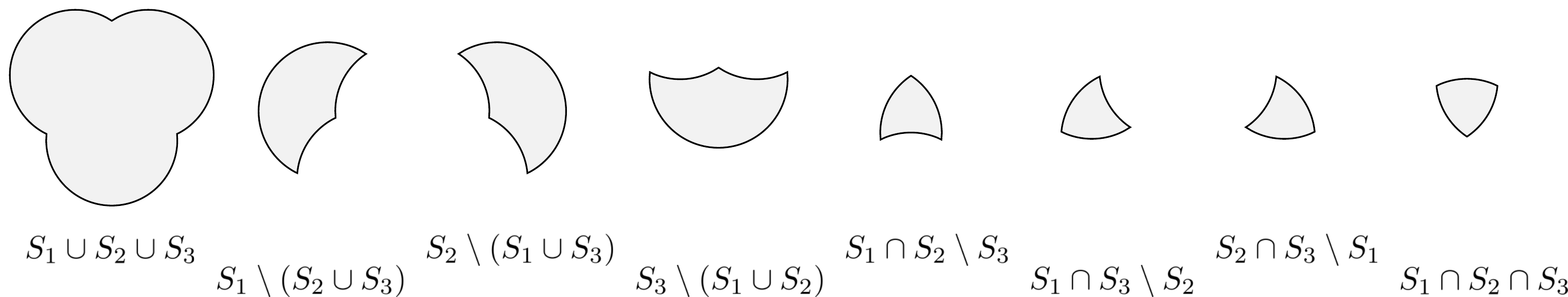}\label{dsft4vis}}
\caption{(a) Generalized coverage function for $N = \{1,2,3\}$ and $c = s_\emptyset = 0$. The shaded area represents the weight of the set and can be negative. (b) $\sft{3}{\coord{s}}$ contains the weights of intersections of all subsets of the set of $S_i$. (c) $\sft{4}{\coord{s}}$ contains the weights of all fragments.\label{interact}}
\end{figure}

The definition generalizes the well-known concept of coverage functions (a class of submodular functions \cite{Krause:14,Chakrabarty:12}) in two ways: it allows negative weights and an offset $c\neq 0$.

Clearly, every generalized coverage function is a set function on $N$. Interestingly, the converse is also true and thus offers a different viewpoint towards set functions. Using this viewpoint we will derive intuitive explanations of the DSFTs and the notion of spectrum for set functions. We postpone the proof as it will be a consequence of the following investigations.

\begin{theorem}\label{setfcts}
Every set function is a generalized coverage function with a suitable chosen $U$ of size $|U|\leq 2^{|N|}$.
\end{theorem}

To interpret Definition~\ref{gencovdef} and give intuition, consider such a generalized coverage function $\coord{s}$ with $c=0$ for simplicity. We can think of $\coord{s}$ assigning a value (or cost) to every subset $A\subseteq N$ of items in $N$. The value of one item $i$ is $s_{\{i\}} = \coord{w}(S_i)$ and the value of a set of two items $i,j$ is $s_{\{i,j\}} = \coord{w}(S_i\cup S_j)$ (see Fig.~\ref{coveragefct}).
In the simplest case, the items are independent, i.e., $s_{\{i,j\}} = s_{\{i\}} + s_{\{j\}} $, meaning the weight of the intersection $\coord{w}(S_i\cap S_j) = 0$. Alternatively, it could be larger than the sum (called complementarity) or smaller than the sum (called substitutability), if the weight of the intersection is positive or negative, respectively. In other words, in a sense, the Venn diagram and weight function capture these interactions in value between two items and among larger sets of items. Because of Theorem~\ref{setfcts}, this interpretation can be applied to every set function. Further, as visualized in Fig.~\ref{interact}, we show next that the weights of these interactions correspond exactly to our notions of spectrum for models 3 and 4.

\subsection{DSFT, type~3}

We will show that for a generalized coverage function $\coord{s}$ in \eqref{gencov}, the $\DSFTt{3}$ computes the weights of the intersections of any subset of the sets $S_i$ as visualized in Fig.~\ref{dsft3vis}. There is one such intersection for any $B\subseteq N$, namely
$$
\bigcap_{j\in B}S_j.
$$
To show this, we apply the inclusion-exclusion principle for the union of sets \cite[p.~158]{Aigner:79} similarly to \cite{Chakrabarty:12}:
\begin{eqnarray}
s_A & = & c + \coord{w}\left(\bigcup_{i\in A}S_i\right) \\ \nonumber
& = & s_\emptyset + \sum_{B\subseteq A, B\neq\emptyset}(-1)^{|B|-1}\coord{w}\left(\bigcap_{j\in B}S_j\right) \label{inclexcl}
\end{eqnarray}
Comparing to the inverse $\DSFTt{3}$ in Table~\ref{tab:smfreq} yields the result.
\begin{theorem}\label{dsft3understand}
If $\coord{s}$ is a generalized coverage function, then
$$
\sft{3}{s}_B = 
\begin{cases}
- \coord{w}\left(\bigcap_{i\in B}S_i\right), & B\neq\emptyset,\\ 
s_\emptyset, & B=\emptyset.
\end{cases}
$$
\end{theorem}

\subsection{DSFT, type~4}

Similarly, but different, the $\DSFTt{4}$ computes the weights of all $2^{n}-1$ disjoint {\em fragments} from which the Venn diagram in Fig.~\ref{coveragesit} is composed. There is one such fragment for every $B\subseteq N$, $B\neq\emptyset$, namely
\begin{equation}\label{fragment}
T_B = \bigcap_{i\in B}S_i\setminus\bigcup_{i\not\in B}S_i.
\end{equation}
The fragments $\{T_B: B \subseteq N, B \neq \emptyset\}$ partition $\bigcup_{i \in N} S_i$, i.e., $\bigcup_{i \in N} S_i = \bigcup_{B \subseteq N, B \neq \emptyset} T_B$ and $T_{B_1} \cap T_{B_2} = \emptyset$ for $B_1 \neq B_2$. Consequently, we can also write each $\bigcup_{i \in A} S_i$,  $A\neq\emptyset$, as the disjoint union $\bigcup_{B \subseteq N, A \cap B \neq \emptyset} T_B$. Using the additivity of $\coord{w}$, we get for all $A\neq\emptyset$
\begin{equation}
\begin{aligned}
s_A &= c + \coord{w}\left(\bigcup_{i \in A} S_i\right) \\
&= c + \sum_{B \subseteq N, A \cap B \neq \emptyset} \coord{w}(T_B) \\
&= \underbrace{c + \sum_{B \subseteq N, B \neq \emptyset} \coord{w}(T_B)}_{= s_N = \widehat{s}_{\emptyset}^{(4)}} - \sum_{B \subseteq N, A \cap B = \emptyset, B \neq \emptyset} \coord{w}(T_B)
\end{aligned}
\end{equation}
Comparing to Table~\ref{tab:smfreq} shows that this equation is, up to the sign, the inverse $\DSFTt{4}$, which yields the result.
\begin{theorem}\label{dsft4understand}
Let $s$ be a generalized coverage function with prior notation. Then
$$
\sft{4}{s}_B = 
\begin{cases}
-\coord{w}\left(\bigcap_{i\in B}S_i\setminus\bigcup_{i\not\in B}S_i\right), & B\neq\emptyset,\\
s_N, & B=\emptyset.
\end{cases}
$$
\end{theorem}

\subsection{Low and High Frequencies: Intuition}

As just shown, the Fourier coefficients $\ft{s}_B$ of a signal in models 3 and 4 capture the interactions between the values of single items $s_{\{i\}}$, $i\in A$ when composed to the value $s_A$. Higher-order interactions (interactions of more items) correspond to larger $B$, lower-order interactions (interactions of few items) to smaller $B$. This motivates again the frequency ordering by the size of $B$ suggested 
in Section~\ref{freqorder} and the following definition.

\begin{definition}\label{bandlim}
We call a set function $m$-band-limited, $m\leq n$, with respect to model~$k$, $k\in\{1,\dots, 5\}$, if
$$
\sft{k}{s}_B = 0\text{ for } |B| > m.
$$
\end{definition}

\begin{figure}\centering
\subfigure[$m=1$]{\includegraphics[scale=0.25]{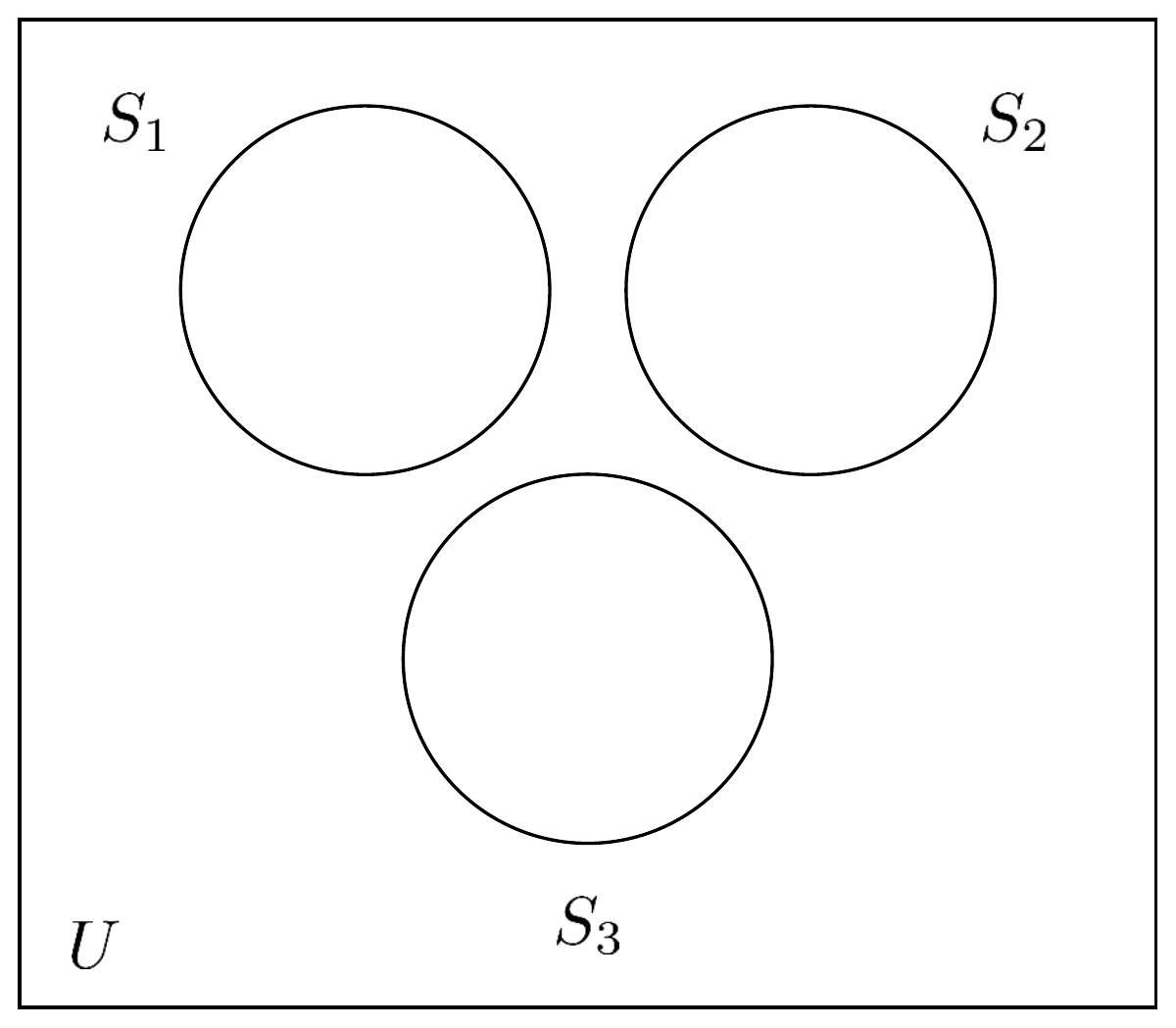}\label{bandlim1}}
\subfigure[$m=2$]{\includegraphics[scale=0.25]{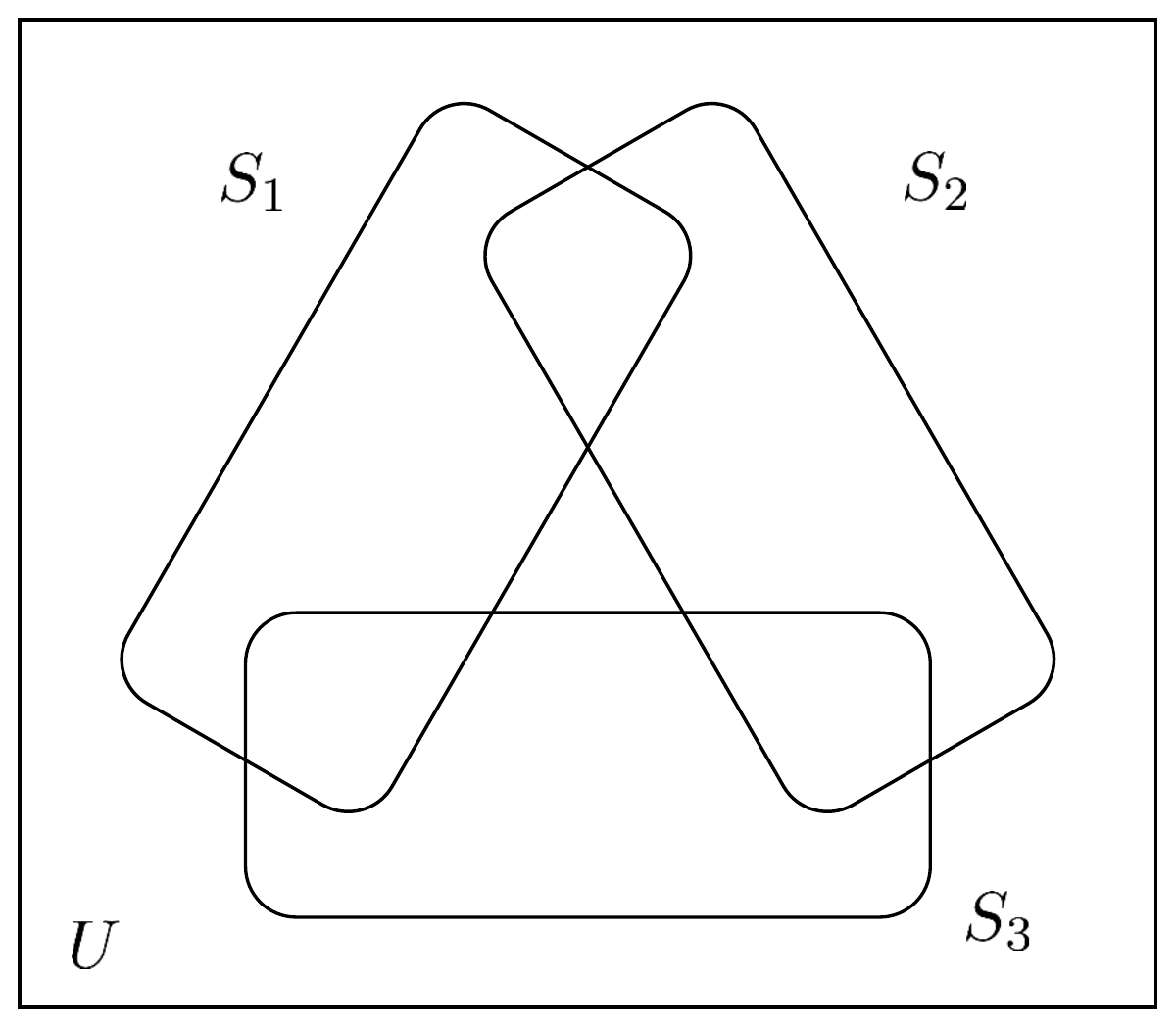}\label{bandlim2}}
\caption{Band-limited set functions w.r.t.~models 3 and 4 for $n=3$ when viewed as generalized coverage functions. The change of shape for $m=2$ is done to better visualize the intersections.\label{fig:bandlim}}
\end{figure}

Fig.~\ref{fig:bandlim} shows the special case of Fig.~\ref{coveragesit} for band-limited set functions for $n=3$ w.r.t.~models 3 and 4. In this case the notions coincide for $m=1,2$. That is, the vector spaces of $1$- and $2$-band-limited set functions w.r.t.~models 3 and 4 coincide. For general $n$, the notion of $m$-band-limited coincides for $m = n-1$ (only the intersection of all $S_i$ has weight $0$), $m = 1$ (no interaction between the $S_i$), and $m=0$ (set function is constant).

For $m=1$, \eqref{gencov} specializes to
$$
s_A = c + \sum_{i\in A}w(i),
$$ 
i.e., it is ($c$ plus) the sum of the weights of its elements. Such a function is called modular because it then satisfies for all $A\subseteq N$, $x,y\in N$, $s_{A\cup\{x\}} + s_{A\cup\{y\}} = s_{A\cup\{x,y\}} + s_A$, the simplest case of a submodular function \cite{Cunningham:85}. 

\subsection{Proof of Theorem~\ref{setfcts}}

The proof of Theorem~\ref{setfcts} is constructive and follows directly from Theorem~\ref{dsft4understand}. Namely, consider a given set function $\coord{s}$ and its spectrum $\sft{4}{\coord{s}}$ in model~4. Construct a Venn diagram of $n$ sets $S_i$ such that each fragment \eqref{fragment} corresponding to a subset $B\subseteq N$ contains one element with weight $-\sft{4}{s}_B$. One element $x$ is outside all $S_i$ with weight $s_N$ and completes the universe $U = \{x\}\cup \bigcup_{i\in N} S_i$. The generalized coverage function $\coord{s}'$ defined this way satisfies by construction $\sft{4}{\coord{s}'} = \sft{4}{\coord{s}}$, i.e., $\coord{s}' = \coord{s}$ as desired.

\subsection{Example: Multivariate Entropy}\label{multivariateentropy}

We consider a random vector $\coord{X}_N = (X_1,\dots,X_n)$ with a joint probability distribution. $X_i$ are random variables and $N = \{1,\dots,n\}$. We consider the set function \cite{Krause:14}
\begin{equation}\label{jent}
\coord{s}:\ 2^N\rightarrow \R,\ A\mapsto H(\coord{X}_A),
\end{equation}
where $H$ is the Shannon entropy and $\coord{X}_A$ is the random vector collecting all $X_i, i\in A$. We will show that the DSFTs type~3 and 4 compute the mutual information structure of the $X_i$. This follows directly from the measure-theoretic Venn diagram interpretation of these concepts shown, e.g., in \cite[pp.~108]{Reza:61}, which exactly matches Fig.~\ref{coveragesit}.

Bivariate mutual information is computed as $I(X;Y) = H(X)+H(Y)-H(X,Y)$. Its multivariate generalization \cite[pp.~57]{Fano:61} is defined recursively as
\begin{multline}
I(X_1;\dots;X_k) = I(X_1;\dots;X_{k-1}) \\ - I(X_1;\dots ;X_{k-1}\mid X_k).
\end{multline}
We will write $I(\coord{X}^\coord{;}_A)$ to denote the mutual information of the random variables in $\coord{X}_A$. 
A formula for computing it directly from the joint entropies is given, e.g., in \cite{Ball:17} and shows that
$$
\sft{3}{s}_B = -I(\coord{X}^\coord{;}_B).
$$
Similarly, from \cite{Bell:03}, we obtain
$$
\sft{4}{s}_B = -I(\coord{X}^\coord{;}_B\mid\coord{X}_{N\setminus B}),\quad B\neq\emptyset,
$$
and $\sft{4}{s}_\emptyset = H(\coord{X}_N)$.

Thus, in a sense, the DSFT of type~3 and 4 generalize the concept of mutual information from the special case of joint entropy (a subclass of the class of submodular set functions \cite{Krause:14}) to all set functions.

\section{Applications}

With the basic SP tool set in place many standard SP algorithms and applications can be ported to the domain of set functions including compression, subsampling, denoising, convolutional neural nets, and others. As mentioned in the introduction, one main challenge with set functions is their large dimensionality, which means that often only few values are available, obtained through an oracle or model that may itself be expensive. In this section we provide two prototypical examples to show how SP techniques may help.

\subsection{Compression}

Compression in its most basic form approximates a signal with its low frequency components as done, e.g., with the DCT in JPEG image compression \cite{Rao:90}. Translated to discrete-set SP, we can approximate a given set function $\coord{s}$ by an $m$-band-limited set function (Definition~\ref{bandlim}) $\coord{s}'$ by dropping the high frequencies. Namely, if $\coord{f}^B$ are the Fourier basis vectors
\begin{equation}\label{compressionexample}
s'_A = \sum_{B \subseteq N, |B| \leq m} \widehat{s}_B f_A^B.
\end{equation}
The chosen DSFT could be for any of the five models. Each $s'_A$ can be computed in $O(k)$ operations, where $k = |\{B\mid |B|\leq m\}| = \sum_{i=0}^m {n\choose i}$. Next we instantiate this idea for a concrete application scenario.

\mypar{Submodular function evaluation} 
Efficient set function representations are of particular importance in the context of submodular optimization, where submodular set functions are minimized or maximized by adaptively querying (i.e., evaluating) set functions \cite{Stobbe:10, Krause:05}. For many practical problems these queries can become a computational bottleneck, e.g., because they involve physical simulations \cite{Krause:14} or require the solution of a linear system of $n$ equations \cite{Krause:08}.

\mypar{Example: Sensor placement}
As an example we consider a set function from a sensor placement task \cite{Krause:05}, in which 46 temperature sensors were placed at Intel Research Berkeley and the goal is to determine the most informative subset of sensors. Formally, let $N = \{1, \dots, 46\}$ and let $X_1, \dots, X_{46}$ be random variables modeling the sensors. The informativeness of a subset of sensors $A \subseteq N$ can be quantified by their joint entropy $H(X_A)$ (equation (2.3) in \cite{Krause:05}), where $X_A = (X_i)_{i \in A}$
(see Section~\ref{multivariateentropy}). In our concrete example, the most informative subset is determined by fitting a multivariate Gaussian model (one variable per sensor) to the sensor measurement data and maximizing the corresponding multivariate entropy $H(X_A)$. 

Because $X_A$ is a multivariate Gaussian random variable we have 
\begin{equation}
s_A = H(X_A) = \frac{1}{2}\log{\det{[K_{i,j}]_{i,j \in A}}} + \frac{|A|}{2} (1 + \log(2\pi)),
\end{equation}
where $K$ is the $n \times n$ covariance matrix and $[K_{i,j}]_{i,j \in A}$ the submatrix corresponding to the sensors in $A$. The evaluation cost of each $s_A$ is in $O(n^3)$ (due to the determinant).  

We reduce the set function evaluation cost to $O(n^2)$, by compressing $s$ with \eqref{compressionexample} to a 2-band-limited function using only the lowest $1 + n + {n \choose 2}$ frequencies in $\mathcal{B} = \{B\mid |B|\leq 2\}$. For model~4, the needed Fourier coefficients can be computed directly in $O(n^2)$, namely (see Table~\ref{tab:smfreq}):
\begin{multline}\label{eq:model4_coefficients}
\widehat s_B^{(4)} = \\ \begin{cases} 
s_N, & B = \emptyset, \\
s_{N \setminus \{x\}} - s_N, & B = \{x\}, \\
s_{N \setminus \{x, y\}} - s_{N \setminus \{x\}} - s_{N \setminus \{y\}} + s_N, & B = \{x, y\}.
\end{cases}
\end{multline}
Using Section~\ref{multivariateentropy}, we can interpret this compression: it ignores high level interactions between sensors, assuming that most of the information is in the entropy of single sensors, and the mutual information of pairs of sensors (Section~\ref{multivariateentropy}).

For comparison we consider model~5 (i.e., the WHT). Note that here each $\widehat s_B^{(5)}$ would require $O(2^n)$ operations and thus cannot be computed exactly. Since the WHT is orthogonal, we can approximate the WHT coefficients using linear regression
\begin{equation}\label{eq:wht_sampling_heuristic}
\widehat{\coord{s}}^{(5)}_{\mathcal{B}} \approx \argmin_{\widehat{\coord{r}}_{\mathcal{B}}} \|\coord{s}_{\mathcal{A}} - \WHT^{-1}_{\mathcal{A}\mathcal{B}}\widehat{\coord{r}}_{\mathcal{B}}\|_2.
\end{equation}
Here, the set $\mathcal{A} = \{A_1, \dots, A_p\}$ consists of $p$ uniformly at random chosen signal indices (subsets), $\coord{s}_{\mathcal{A}}$ is the corresponding subsampled version of $\coord{s}$ and $\WHT^{-1}_{\mathcal{A}\mathcal{B}}$ is the submatrix of $\WHT^{-1}$ obtained by selecting row indices in $\mathcal{A}$ and column indices in $\mathcal{B}$. We consider different values of $p$ in the experiment. \eqref{eq:wht_sampling_heuristic} now finds the best approximating Fourier coefficients in our frequency band $\mathcal{B}$.

\mypar{Results} 
Using \eqref{eq:model4_coefficients} for model~4 and \eqref{eq:wht_sampling_heuristic} for model~5, we can now compute approximate set function values $s'_C$ with \eqref{compressionexample} for any $C\subseteq N$.

In Table \ref{tab:compression_full}, we compare the compression quality of models~4 and 5 in terms of the approximate relative reconstruction error 
\begin{equation}\label{eq:relative_error_approximate}
E'_m(\coord{s}, \coord{s}') = ||\coord{s}_\mathcal{C} - \coord{s}'_\mathcal{C}||_2/||\coord{s}_\mathcal{C}||_2,
\end{equation}
where the set $\mathcal{C}\subseteq 2^N$ consists of $m = 10^6$ randomly chosen signal indices (subsets). $E'_m(\coord{s}, \coord{s}')$ converges to the actual error as $m\rightarrow\infty$.

The table shows that model~4 approximates well in this case, far superior to model~5.

\begin{table}\centering
\caption{Approximation error for a submodular function associated with sensor data. For the random regressions (WHT) we report mean and standard deviations over the randomness of $\mathcal{A} = \{A_1, \dots, A_p\}$ in 20 runs.\label{tab:compression_full}}
\begin{small}
\begin{tabular}{@{}ll@{}}\toprule
Method & $E'_{10^6}(\coord{s}, \coord{s}')$ \\ \midrule
DSFT, type~4 & $0.00523$\\
WHT random regression $p = 10^3$ & $0.43107 \pm 0.06489$ \\
WHT random regression $p = 10^4$ & $0.29293 \pm 0.00016$\\
WHT random regression $p = 10^5$ & $0.29288 \pm 0.00006$ \\ \bottomrule
\end{tabular}
\end{small}
\end{table}

\subsection{Sampling} 

We derive a novel sampling strategy for set functions that are $k$-sparse in the Fourier domain and present a potential application in the domain of auction design.

\mypar{Sampling theorem}
Consider a Fourier sparse set function $s$ with known Fourier support $\mbox{supp}(\widehat s) = \{B_1, \dots, B_k\} = \mathcal{B}$. Sampling theorems address the question of computing the associated Fourier coefficients $\widehat{\coord{s}}_{\mathcal{B}}$ from few (typically also $k$) queries, i.e., set function evaluations. Following the sampling theory from \cite{Vetterli:14}, the basic task is to select $k$ subsets $\mathcal{A} = \{A_1, \dots, A_k\}$ such that the linear system of equations 
\begin{equation}
s_A = \sum_{B \in \mathcal{B}} \widehat s_{B} f^{B}_A \mbox{ for } A \in \mathcal{A}
\end{equation}
has a unique solution. Equivalently, this is the case if and only if the submatrix
$
(\DSFT^{-1})_{\mathcal{A}\mathcal{B}}
$
is invertible. The choice of the sampling indices $\mathcal{A}$ thus depends on the type~(1--5) of $\DSFT$. Here we consider type~4.

\begin{theorem}\emph{(Model~4 Sampling)}\label{thm:dsft4_sampling}
Let $s$ be a Fourier sparse set function with Fourier support $\mbox{supp}(\widehat s) = \{B_1, \dots, B_k\} = \mathcal{B}$. Let $\mathcal{A} = \{N\setminus B_1, \dots, N\setminus B_k\}$. Then
$$
T = ((\DSFTt{4})^{-1})_{\mathcal{A}\mathcal{B}}
$$
is invertible, i.e., $s$ can be perfectly reconstructed from its queries at $\mathcal{A}$:
$$
\coord{s} = \Big(((\DSFTt{4})^{-1})_{2^N\mathcal{B}}T^{-1}\Big)\coord{s}_\mathcal{A}.
$$
\end{theorem}
\begin{proof}
The matrix $(\DSFTt{4})^{-1}$ has an upper left triangular shape (Table~\ref{tab:smfreq}), i.e., its diagonal elements have indices $(B, N\setminus B)$. Thus, $((\DSFTt{4})^{-1})_{\mathcal{A}\mathcal{B}}$ is also upper left triangular and thus invertible, as desired.
\end{proof}
If a set function is only approximately Fourier-sparse, i.e., most Fourier coefficients are very small, one can use Theorem~\ref{thm:dsft4_sampling} for an approximate reconstruction. We now present a possible application: preference elicitation in combinatorial auctions.

\mypar{Example: Auction design}
In combinatorial auctions \cite{Parkes:06} a set of goods $N = \{1, \dots, n\}$ is sold to a set of bidders $M = \{1, \dots, m\}$. Every bidder $i\in M$ is modeled as a set function (called value function) $v^i: 2^N \to \R^{\geq 0}$, which
assigns to every bundle (subset) of goods their value for bidder $i$. The goal of an auction is to find an efficient allocation of the goods to the bidders. In order to do so, the so-called social welfare function 
\begin{equation}
V(A_1, \dots, A_m) = \sum_{i = 1}^m v^i_{A_i}
\end{equation}
is maximized over all possible allocations of items to the bidders, i.e., all $(A_1, \dots, A_m)$ with $A_i\subseteq N$ and $A_i \cap A_j = \emptyset$ for $i \neq j$. One major difficulty arises from the fact that the true value functions $v^i$ are unknown to the auctioneer and can only be accessed through a limited number of queries called preference elicitation \cite{Brero:18}. As querying bidders in a real world auction amounts to asking them to report their values for certain subsets of goods, the maximum number of queries per bidder is typically capped by 500.  

Machine learning based preference elicitation approaches overcome this issue by approximating the value functions by parametric functions, e.g., polynomials of degree two \cite{Brero:18} or Gaussian processes \cite{Brero:19}. The estimated parameters of these approximations are adaptively refined using a suitable querying strategy. We propose to apply our sampling theorem to determine the queries and approximations of the $v^i$.

\begin{figure}\begin{center}
\subfigure{\includegraphics[width=0.24\textwidth]{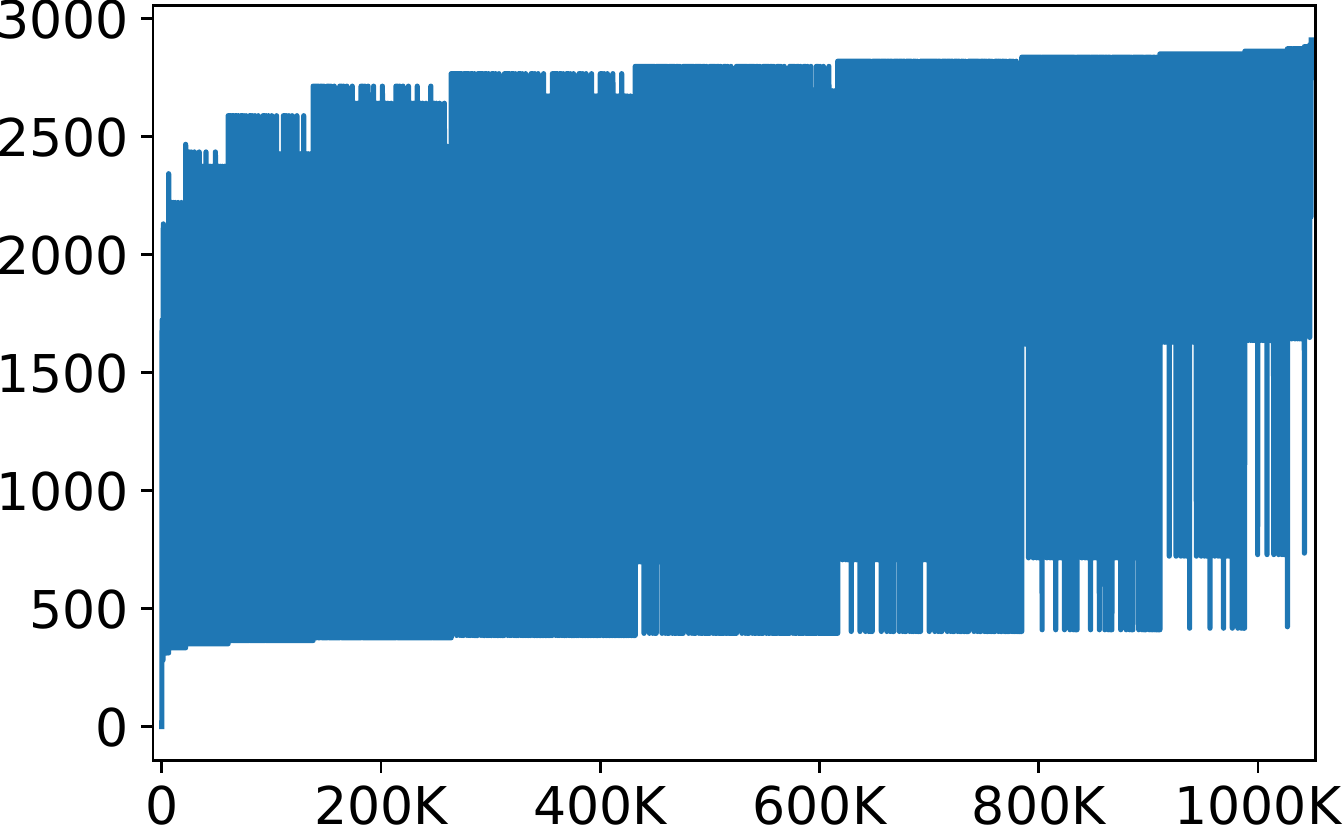}}\subfigure{\includegraphics[width=0.24\textwidth]{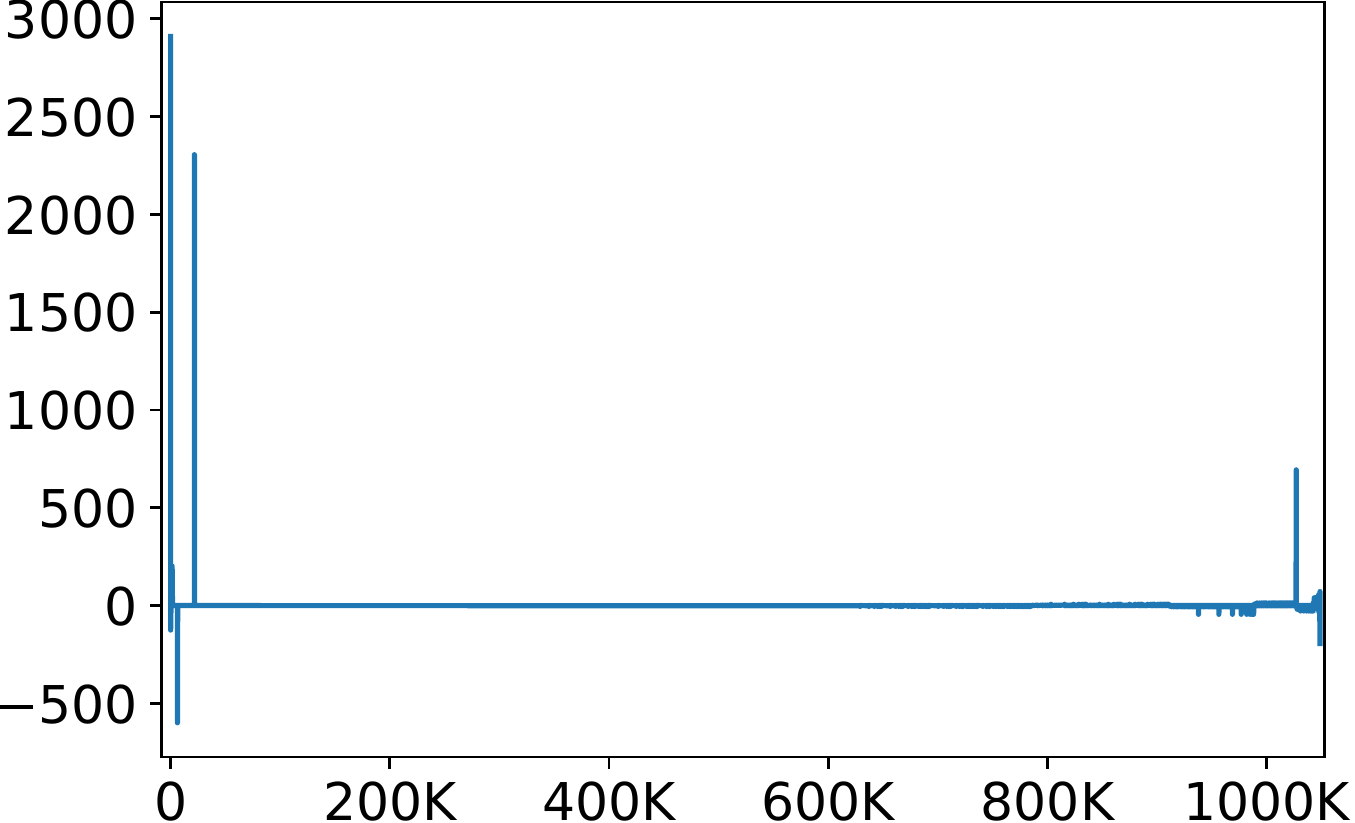}}

\subfigure{\includegraphics[width=0.24\textwidth]{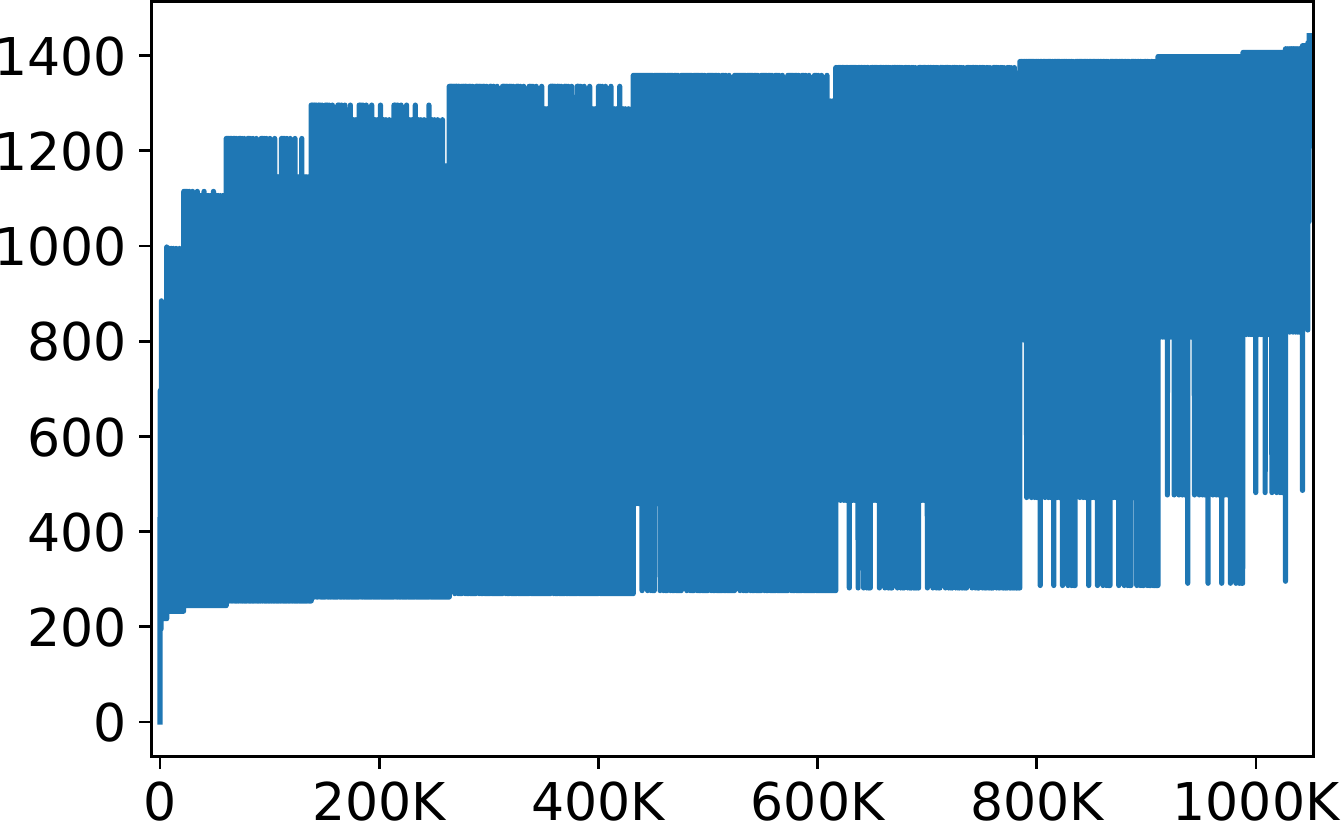}}\subfigure{\includegraphics[width=0.24\textwidth]{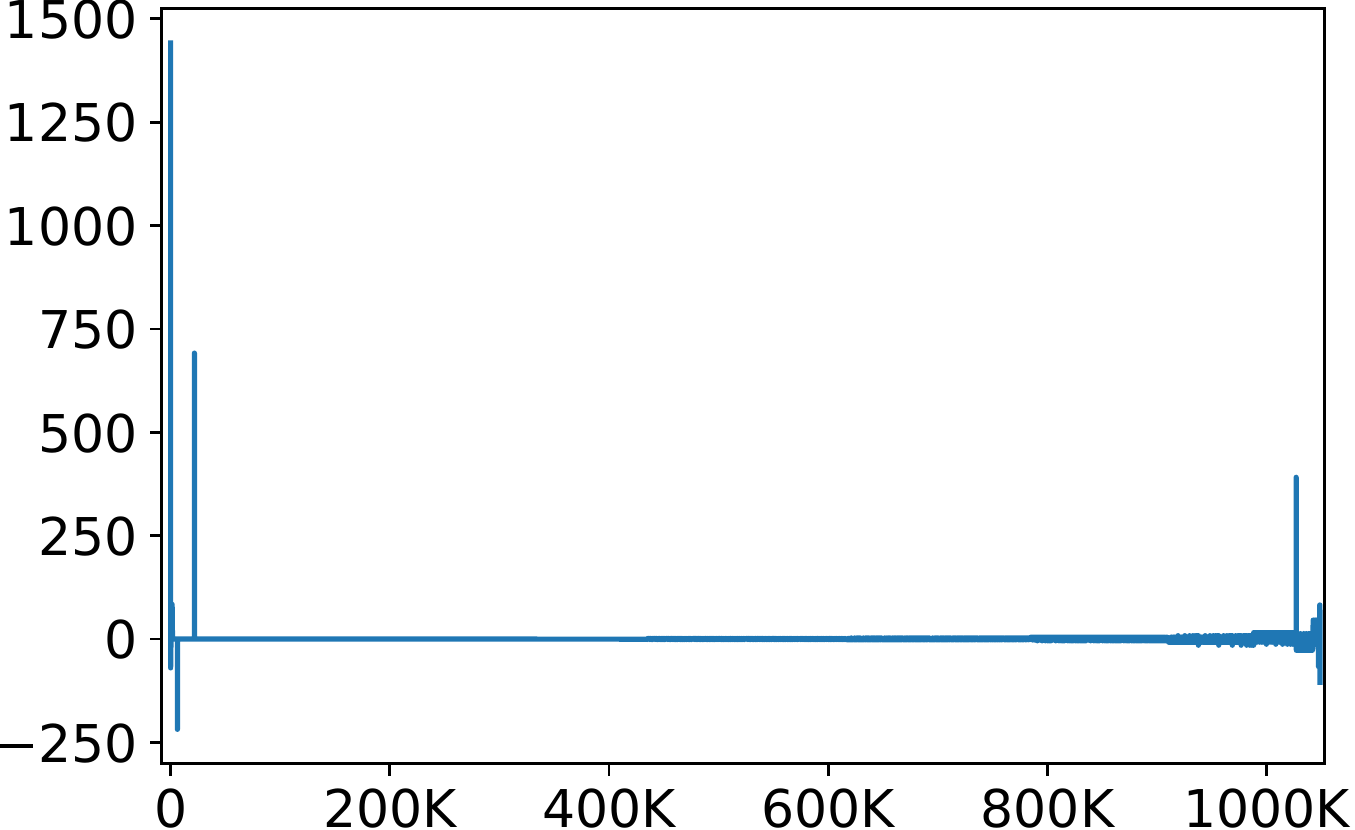}}

\subfigure{\includegraphics[width=0.24\textwidth]{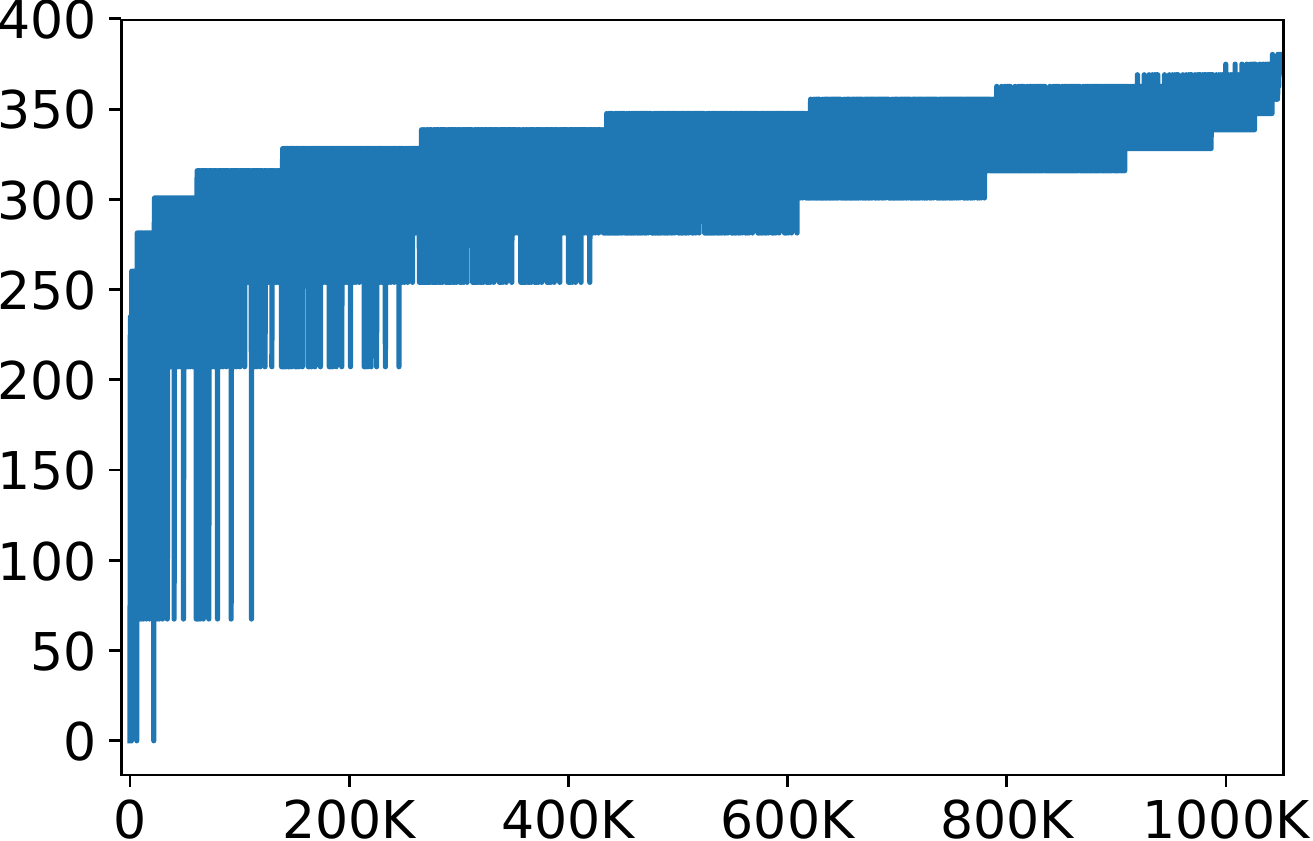}}\subfigure{\includegraphics[width=0.24\textwidth]{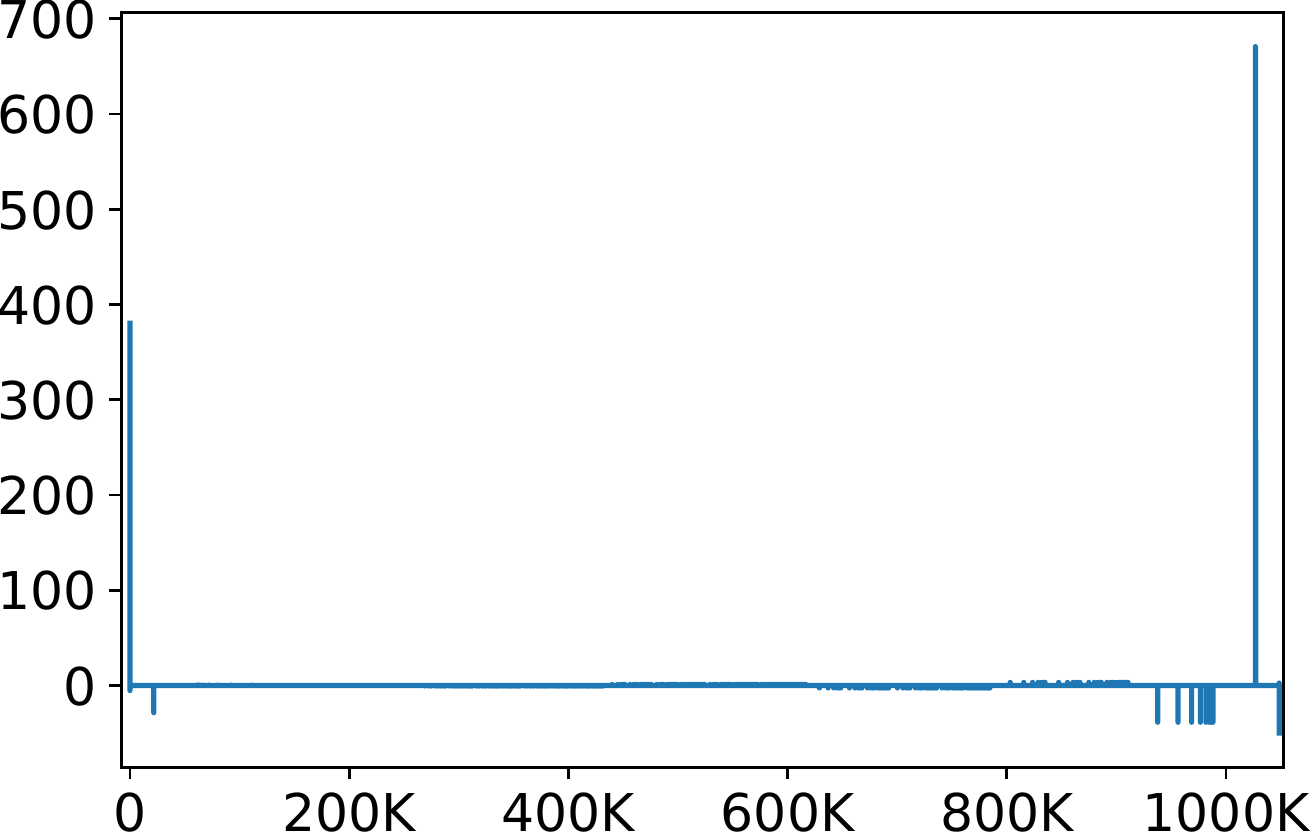}}
\caption{Example value functions (first column) and respective Fourier transforms (second column) of different SRVM bidder types (rows). The subsets on the x-axis are ordered by size.}\label{fig:sampling}
\end{center}
\end{figure}

To assess the viability of our approach we consider spectrum auctions and the single region valuation model (SRVM) \cite{Kroemer:16} (one of several models commonly used in the research of spectrum auctions) to generate goods and bidders. Concretely, we generate a country with 3 frequency bands and 20 associated goods (licenses for the frequency bands) to be auctioned. There are 3 bidder types in the model. Each bidder's value function is defined in terms of some bidder type specific parameters, which are randomly sampled to obtain example bidders. In Figure~\ref{fig:sampling} we plot the value functions and their Fourier transforms for example bidders of the three different types. Observe that most Fourier coefficients (second column) are zero.

For each type $t \in \{1, 2, 3\}$ we generate 50 random bidders $v^t_1, \dots, v^t_{50}$. For convenience we write $v = v^t_j$. We use 25 for training: we compute their type~4 spectra $\widehat{v}$ to determine the 500 (out of $2^{20}$) most important frequency locations $\mathcal{B}$ and then use Theorem~\ref{thm:dsft4_sampling} to determine the associated 500 signal indices $\mathcal{A}$ for sampling. Then, we query $\mathbf{v}_{\mathcal{A}}$ for the other 25 bidders (the test set) and use Theorem~\ref{thm:dsft4_sampling} to compute a Fourier sparse approximation $v'$ in each case. If the support of a bidder $v$ in the test set was contained in the determined $\mathcal{B}$, its $v'$ would be equal to $v$. 

\mypar{Results} Table~\ref{tab:sampling_train_test_dollars_srvm} shows mean and standard deviations of the relative reconstruction errors\footnote{Notice that in this example the small ground set allows for the exact computation of the relative reconstruction error, which we had to approximate in our previous experiment in \eqref{eq:relative_error_approximate}.} $\|v - v'\|_2/\|v\|_2$ 
for all 3 types (i.e., $t \in \{1, 2, 3\}$) in comparison to the second-degree polynomial approximation in \cite{Brero:18} based on 500 queries and based on the entire value function. Our sampling strategy based on discrete-set SP yields higher accuracy in the experiment and offers a novel method for preference elicitation in real-world auctions. 

Recently, we extended these ideas to approximate set functions under the assumption of sparse but unknown Fourier support \cite{Wendler:20}. Building on this work, we then proposed an iterative combinatorial auction mechanism that achieves state-of-the-art results in various auction domains \cite{Weissteiner:20}.

\begin{table}\centering
\caption{Relative reconstruction error for 3 different SRVM bidder types (1--3) by querying 500 valuations using Theorem~\ref{thm:dsft4_sampling} and by polynomial approximation based on 500 or all valuations.}\label{tab:sampling_train_test_dollars_srvm}
\begin{small}
\begin{tabular}{@{}llll@{}}\toprule
& DSFT4 500 & poly2 500 & poly2 all\\ \midrule
1 & $0.00037 \pm 0.00019$ & $0.07 \pm 0.003$ & $0.05 \pm 0.002$ \\
2 & $0.00042 \pm 0.00016$ & $0.04 \pm 0.002$ & $0.03 \pm 0.001$  \\
3 & $0.00064 \pm 0.00016$ & $0.05 \pm 0.003$ & $0.04 \pm 0.001$\\ \bottomrule
\end{tabular}
\end{small}
\end{table}

\section{Conclusion}

Signal processing theory and tools have much to offer in modern data science but sometimes require adaptation to be applicable to new types of data that are structurally very different from traditional audio and image signals. In this paper we considered signals on powersets, i.e., set functions, and used algebraic signal processing (ASP) to derive novel forms of discrete-set SP from different definitions of set shifts. Our work brings the basic SP tool set of convolution and Fourier transforms and an SP point of view to the domain of set functions. Using the concept of general coverage function we showed that our notion of spectrum is intuitive: it captures the complementarity and substitutability of items in the ground set, with multivariate mutual information as a special case. Possible applications include the wide area of submodular function optimization in image segmentation, recommender systems, or sensor selection, but also signals on hypergraphs and auction design. In particular, discrete-set SP provides new tools to reduce the dimensionality of set functions through SP-based compression or sampling techniques. We showed two prototypical examples in sensor selection and preference elicitation in auctions.


%

%

\ifCLASSOPTIONcaptionsoff
  \newpage
\fi



\bibliographystyle{IEEEtran}
\bibliography{paper}

\end{document}